\documentclass[smallextended]{svjour3}%
\pdfoutput=1
\usepackage{fullpage}
\usepackage{amsmath}
\usepackage{amsfonts}
\usepackage{amssymb}
\usepackage{graphicx}%
\providecommand{\U}[1]{\protect\rule{.1in}{.1in}}

\journalname{Quantum Information Processing}
\begin{document}

\title{Public and private resource trade-offs for a quantum channel\thanks{M.M.W.
acknowledges support from the MDEIE (Qu\'{e}bec) PSR-SIIRI international
collaboration grant.} }
\author{Mark M. Wilde
\and Min-Hsiu Hsieh}
\institute{Mark M. Wilde is a postdoctoral fellow with the School of Computer Science, McGill University, Montreal, Quebec, Canada H3A 2A7 \email{mark.wilde@mcgill.ca}
\and Min-Hsiu Hsieh was with the
ERATO-SORST Quantum Computation and Information Project, Japan Science and Technology Agency 5-28-3, Hongo, Bunkyo-ku, Tokyo, Japan during the development of this paper. He is now with the Statistical Laboratory, University of Cambridge, Wilberforce Road, Cambridge, UK CB3 0WB \email{minhsiuh@gmail.com}}

\date{\today}
\maketitle

\begin{abstract}
Collins and Popescu realized a powerful analogy between several resources in
classical and quantum information theory. The Collins-Popescu analogy states
that public classical communication, private classical communication, and
secret key interact with one another somewhat similarly to the way that
classical communication, quantum communication, and entanglement interact.
This paper discusses the information-theoretic treatment of this analogy for
the case of noisy quantum channels. We determine a capacity region for a
quantum channel interacting with the noiseless resources of public classical
communication, private classical communication, and secret key. We then
compare this region with the classical-quantum-entanglement region from our
prior efforts and explicitly observe the information-theoretic consequences of
the strong correlations in entanglement and the lack of a super-dense coding
protocol in the public-private-secret-key setting. The region simplifies for
several realistic, physically-motivated channels such as entanglement-breaking
channels, Hadamard channels, and quantum erasure channels, and we are able to
compute and plot the region for several examples of these channels. 

\end{abstract}

\section{Introduction}

One of the first breakthroughs in quantum information theory was the discovery
of a protocol for establishing secret correlations with the use of a quantum
channel~\cite{BB84}. Such a task is now well known as quantum key
distribution~\cite{RevModPhys.81.1301}. This thriving area of research has
resulted in a currently available quantum technology, and efforts are now
underway to construct space-to-ground quantum communication
devices~\cite{U09,M08}.

These initial results on quantum key distribution inspired the quantum
information-theoretic study of secret communication over quantum channels, and
this line of inquiry has subsequently led to an improved understanding of the
relations between private classical information and quantum information.
Schumacher and Westmoreland were one of the first to study this
connection~\cite{SW98}, and Collins and Popescu then discussed a useful
analogy between the classical world and the quantum world~\cite{CP02}. The
Collins-Popescu analogy states that the way that a public classical bit, a
private classical bit, and a bit of secret key interact is qualitatively
similar to the way that a classical bit, a quantum bit, and a bit of
entanglement interact~\cite{CP02}. They justify this analogy operationally, by
comparing the teleportation protocol~\cite{BBCJPW93} to the one-time pad
protocol~\cite{V26}. Teleportation consumes two classical bits and one
maximally entangled pair to generate a qubit channel, whereas the one-time pad
protocol consumes one public bit and a bit of secret key to establish a
private classical bit. A qubit channel can establish entanglement, and a
private classical bit channel can establish a bit of secret key---these
protocols have the respective names entanglement distribution and secret key
distribution. Additionally, a qubit channel can generate a classical bit, and
a private classical bit channel can generate a public classical
bit.\footnote{This latter protocol, that we call private-to-public
transmission, follows from the particular communication model that we consider
in this paper.} But the lack of an analogy of the super-dense coding
protocol~\cite{BW92} in the public-private-secret-key setting is where this
analogy breaks down.

Shortly after this initial work, Devetak and Cai \textit{et al.}~independently
established the private capacity of a quantum channel as one of its
fundamental capacities~\cite{Devetak03,CWY04}. These results and the ideas
involved are formally similar to private information transmission in the
classical setting~\cite{CK67,Wyner75Bell,AC93CR,M93,AC93II}. In addition to
determining the private classical capacity, Devetak provided a good lower
bound on the quantum capacity of a quantum channel by showing how to construct
good quantum error-correcting codes from classical codes that transmit
classical information privately~\cite{Devetak03}. Devetak and Winter continued
these efforts, demonstrating many further important connections between
private classical information and quantum information~\cite{DW03b,DW03c}, and
Smith\textit{et al.}~then employed these capacity formulas to determine good
bounds on the secret key rate of the standard protocol for quantum key
distribution~\cite{PhysRevLett.100.170502}.

In this article, we study how a noisy quantum channel interacts with the
noiseless resources of public classical communication, private classical
communication, and secret key. That is, we determine trade-off formulas for
how a sender and receiver can use any of the noiseless resources to assist a
noisy quantum channel in generating any of the other noiseless resources. In
earlier work, we determined trade-off formulas in the setting where the
noiseless resources are classical communication, quantum communication, and
entanglement~\cite{HW08ITIT,HW09T3,HW09book,WH10}. Thus, one could view the
present work as the completion of the information-theoretic treatment of the
Collins-Popescu analogy (at least for the case of channels) that began in the
aforementioned papers and continued in Refs.~\cite{DS03,hsieh:042306}.

Our main result is the private dynamic capacity theorem. This theorem
determines the capabilities for a noisy quantum channel to generate any of the
three noiseless classical resources when assisted by the others. The rates in
the private dynamic capacity region can be either positive or negative,
depending on whether a protocol generates or consumes a given resource,
respectively. The result of this theorem is that combinations of only four
protocols are sufficient to generate the entire capacity region:\ the
publicly-enhanced private father protocol~\cite{HW09}, the one-time pad
protocol, secret key distribution, and private-to-public transmission. This
result is in line with the Collins-Popescu analogy because we found that the
classically-enhanced father protocol~\cite{HW08ITIT}, teleportation,
entanglement distribution, and super-dense coding are sufficient to realize
the quantum dynamic capacity region of a quantum
channel~\cite{HW08ITIT,HW09T3,HW09book,WH10}. This theorem also explicitly
demonstrates the aforementioned breakdown of the Collins-Popescu analogy---the
last two inequalities in each theorem are similar by inspection, but the first
one in each is different because of the lack of a super-dense coding protocol
in the public-private-secret-key setting and because the rates in
teleportation and the one-time pad are different.

We also explicitly compute and plot the private dynamic capacity region for
several realistic, physically motivated quantum
channels:\ entanglement-breaking channels~\cite{shor:4334,HSR03}, dephasing
channels, cloning channels~\cite{BHP09,B09,BDHM09,BHTW10}, and erasure
channels~\cite{PhysRevA.56.33}. Entanglement-breaking channels have
application in entanglement detection
protocols~\cite{PhysRevA.78.062105,PhysRevA.80.062314}. Dephasing noise occurs
in superconducting qubit systems \cite{BDKS08}, the cloning channel represents
a natural process that occurs during stimulated
emission~\cite{MH82,SWZ00,LSHB02}, and the erasure channel is a simplified
model for photon
loss~\cite{PhysRevLett.91.217901,PhysRevA.75.042316,Lu:2008:11050,DGJZ10}.
Br\'{a}dler \textit{et al.}~pointed out in Ref.~\cite{BHTW10} that both
dephasing channels and cloning channels are examples of Hadamard
channels~\cite{KMNR07}, and this Hadamard property is useful in proving that
the private dynamic capacity region is tractable. The proof for the quantum
erasure channel follows by exploiting its particular structure. We prove these
results first by showing that a formula, named the private dynamic capacity
formula, is additive for each of these channels. We then analyze each channel
individually and show that a particular ensemble suffices to achieve the
boundary points of the private dynamic capacity region.

We structure this paper as follows. We first review the communication model,
some definitions, and notation that are essential in understanding the rest of
the paper. Section~\ref{sec:capacity-theorem}\ states the private dynamic
capacity theorem and the next two sections prove the achievability part and
the converse part. We then introduce the private dynamic capacity formula,
show how its additivity implies that the computation of the capacity region
boundary simplifies, analyze special cases of
the formula, and compare the region to the quantum dynamic capacity region
from Refs.~\cite{HW09T3,HW09book,WH10}. Sections~\ref{sec:EB}\ and
\ref{sec:Hadamard}\ prove that the private dynamic capacity formula is
additive for entanglement-breaking channels and the Hadamard class of
channels, respectively. We finally compute and plot the private dynamic
capacity region for dephasing channels, cloning channels, and erasure channels
in Section~\ref{sec:special-channels}. We conclude with a discussion and some
open problems.

\section{Definitions and notation}

\label{sec:def-not}We first establish some definitions and notation that we
employ throughout the paper, and we review a few important properties of the
entropy. Consider a random variable $M$ with a uniform distribution on $D$
values. Let $\overline{\Phi}^{M_{A}M_{B}}$ denote an embedding of this random
variable into a maximally correlated state shared between two parties~$M_{A}$
and$~M_{B}$:%
\begin{equation}
\overline{\Phi}^{M_{A}M_{B}}\equiv\frac{1}{D}\sum_{m=1}^{D}\left\vert
m\right\rangle \left\langle m\right\vert ^{M_{A}}\otimes\left\vert
m\right\rangle \left\langle m\right\vert ^{M_{B}}.
\label{eq:max-correlated-state}%
\end{equation}
A common randomness bit corresponds to the special case where $D=2$. Suppose a
third party Eve possesses a quantum system $E$. A state $\rho^{M_{A}M_{B}E}$
on the systems $M_{A}$, $M_{B}$, and $E$ is a public common randomness state
if%
\begin{align*}
\text{Tr}_{E}\{\rho^{M_{A}M_{B}E}\}  &  =\overline{\Phi}^{M_{A}M_{B}},\\
\rho^{M_{A}M_{B}E}  &  \neq\overline{\Phi}^{M_{A}M_{B}}\otimes\sigma^{E},
\end{align*}
for some state $\sigma^{E}$. The above conditions imply that Eve has some
correlations with the above state and could learn about the random variable
$M$ by performing a measurement on her system. A state $\omega^{M_{A}M_{B}E}$
is a secret key state if%
\begin{align*}
\text{Tr}_{E}\{\omega^{M_{A}M_{B}E}\}  &  =\overline{\Phi}^{M_{A}M_{B}},\\
\omega^{M_{A}M_{B}E}  &  =\overline{\Phi}^{M_{A}M_{B}}\otimes\sigma^{E},
\end{align*}
for some state $\sigma^{E}$. In this case, Eve cannot learn anything about the
random variable $M$ by performing a measurement on her share of $\omega
^{M_{A}M_{B}E}$.

A completely-positive trace-preserving (CPTP)\ map $\mathcal{N}^{A^{\prime
}\rightarrow B}$\ is the most general map we consider that maps from a quantum
system $A^{\prime}$\ to another quantum system $B$\ (we usually call them
\textquotedblleft Alice\textquotedblright\ and \textquotedblleft
Bob\textquotedblright). It acts as follows on any density operator $\rho$:%
\[
\mathcal{N}^{A^{\prime}\rightarrow B}\left(  \rho\right)  =\sum_{k}A_{k}\rho
A_{k}^{\dag},
\]
where the operators $A_{k}$ satisfy the condition $\sum_{k}A_{k}^{\dag}%
A_{k}=I$. A quantum channel admits an isometric extension $U_{\mathcal{N}%
}^{A^{\prime}\rightarrow BE}$, which is a unitary embedding into a larger
Hilbert space. One recovers the original channel by taking a partial trace
over the \textquotedblleft environment\textquotedblright\ system $E$ (we
usually call this system \textquotedblleft Eve\textquotedblright). One obtains
the complementary channel~$\left(  \mathcal{N}^{c}\right)  ^{A^{\prime
}\rightarrow E}$ by taking a partial trace over the system $B$.

A channel is degradable if there is a degrading map~$\mathcal{D}^{B\rightarrow
E}$\ such that Bob can simulate the map to Eve~\cite{DS03}:%
\[
\forall\rho\ \ \ \ \ \ \mathcal{D}^{B\rightarrow E}\circ\mathcal{N}%
^{A^{\prime}\rightarrow B}\left(  \rho\right)  =\left(  \mathcal{N}%
^{c}\right)  ^{A^{\prime}\rightarrow E}\left(  \rho\right)  .
\]
A channel is antidegradable if there is a map $\mathcal{T}^{E\rightarrow B}$
such that Eve can simulate the map to Bob:%
\[
\forall\rho\ \ \ \ \ \ \mathcal{T}^{E\rightarrow B}\circ\left(  \mathcal{N}%
^{c}\right)  ^{A^{\prime}\rightarrow E}\left(  \rho\right)  =\mathcal{N}%
^{A^{\prime}\rightarrow B}\left(  \rho\right)  .
\]
A channel $\mathcal{N}_{\text{EB}}^{A^{\prime}\rightarrow B}$ is
entanglement-breaking if its output is a separable state whenever the input is
entangled~\cite{shor:4334,HSR03}:%
\[
\mathcal{N}_{\text{EB}}^{A^{\prime}\rightarrow B}(\left\vert \Gamma
\right\rangle \left\langle \Gamma\right\vert ^{AA^{\prime}})=\sum_{x}%
p_{X}\left(  x\right)  \sigma_{x}^{A}\otimes\theta_{x}^{B}.
\]
Such a channel is antidegradable and the antidegrading map consists of two
parts:\ 1)\ a measurement of the system $E$ that gives a classical variable
and 2) a state preparation conditional on the classical outcome of the
measurement. A quantum Hadamard channel is one whose complementary channel is
entanglement-breaking~\cite{KMNR07}. It is thus degradable with a similar
degrading map that consists of a measurement and state preparation.%
\begin{figure}
[ptb]
\begin{center}
\includegraphics[
natheight=3.653800in,
natwidth=8.146500in,
height=2.2788in,
width=5.047in
]%
{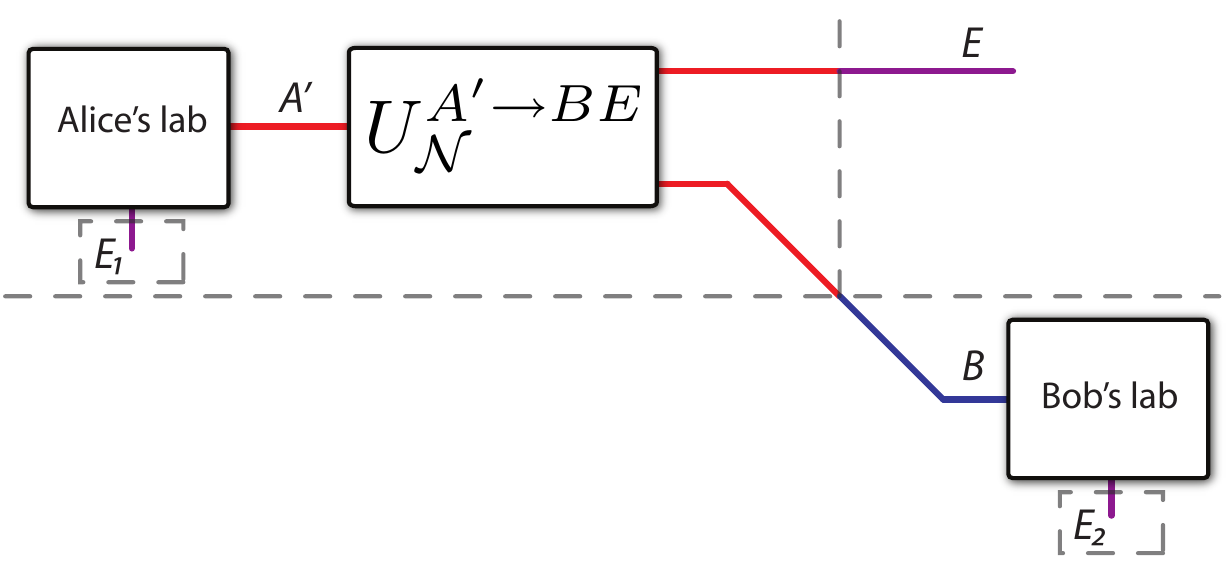}%
\caption{(Color online) The communication model in this paper. Alice can
prepare local states in her lab and choose to send them through a noisy
channel or dump them locally at no cost in a bin $E_{1}$\ to which Eve has
access. We depict the isometric extension $U_{\mathcal{N}}^{A^{\prime
}\rightarrow BE}$\ of the channel $\mathcal{N}^{A^{\prime}\rightarrow B}$ and
give Eve full access to the environment~$E$ of the channel. Bob receives the
output~$B$\ of the channel and can process it locally at his end or dump it in
a bin$~E_{2}$\ to which Eve has access. In this model, Alice and Bob can
simulate a public classical channel from a private classical channel because
Bob can choose to dispose quantum states in the bin~$E_{2}$ to which Eve has
access.}%
\label{fig:comm-model}%
\end{center}
\end{figure}

We employ a particular model of communication in this paper (depicted in
Figure~\ref{fig:comm-model}). We define a public channel as one for which an
eavesdropper Eve can gain some information about what Alice and Bob transmit
over it. A private channel is one for which Eve cannot gain any information
about what they transmit. In this model, we give Eve access to the environment
of a noisy quantum channel and particular registers that Alice and Bob can
discard locally from their laboratories. We do not count this discarding as a
resource because it results from local actions that Alice and Bob take. This
particular model allows us to make close contact with result from the
classical-quantum-entanglement trade-off~\cite{HW08ITIT,HW09T3,HW09book,WH10}.

We consider a three-dimensional capacity region throughout this work (as in
Ref.~\cite{HW09T3}), whose points $\left(  R,P,S\right)  $\ correspond to
rates of public classical communication, private classical communication, and
secret key generation/consumption, respectively. For example, the one-time pad
protocol corresponds to the following point:%
\[
\left(  -1,1,-1\right)  ,
\]
because it consumes a public bit and a bit of secret in order to generate a
private bit. Secret key distribution corresponds to%
\[
\left(  0,-1,1\right)  ,
\]
and a private-to-public transmission corresponds to%
\[
\left(  1,-1,0\right)  .
\]

The entropy $H\left(  A\right)  _{\rho}$ of a density operator $\rho^{A}$ on
some quantum system $A$ is as follows:%
\[
H\left(  A\right)  _{\rho}\equiv-\text{Tr}\left\{  \rho^{A}\log\rho
^{A}\right\}  ,
\]
where the logarithm is base two. The entropy can never exceed the logarithm of
the dimension of system $A$. The quantum mutual information of a bipartite
density operator $\rho^{AB}$ is%
\[
I\left(  A;B\right)  _{\rho}\equiv H\left(  A\right)  _{\rho}+H\left(
B\right)  _{\rho}-H\left(  AB\right)  _{\rho},
\]
and the conditional quantum mutual information for a tripartite state
$\rho^{ABC}$ is%
\[
I\left(  A;B|C\right)  _{\rho}=H\left(  AC\right)  _{\rho}+H\left(  BC\right)
_{\rho}-H\left(  C\right)  _{\rho}-H\left(  ABC\right)  _{\rho}.
\]
The quantum mutual information obeys a chain rule:%
\begin{equation}
I\left(  AB;C\right)  _{\rho}=I\left(  A;C\right)  _{\rho}+I\left(
B;C|A\right)  _{\rho}. \label{eq:mut-chain-rule}%
\end{equation}

A classical-quantum state $\sigma^{XYBE}$\ of the following form plays an
important role throughout this paper:%
\[
\sigma^{XYBE}\equiv\sum_{x,y}p_{X,Y}\left(  x,y\right)  \left\vert
x\right\rangle \left\langle x\right\vert ^{X}\otimes\left\vert y\right\rangle
\left\langle y\right\vert ^{Y}\otimes U_{\mathcal{N}}^{A^{\prime}\rightarrow
BE}(\rho_{x,y}^{A^{\prime}}),
\]
where the states $\rho_{x,y}^{A^{\prime}}$ are mixed states and
$U_{\mathcal{N}}^{A^{\prime}\rightarrow BE}$ is the isometric extension of
some noisy channel $\mathcal{N}^{A^{\prime}\rightarrow B}$. Applying the above
chain rule gives the following relation:%
\begin{equation}
I\left(  YX;B\right)  _{\sigma}=I\left(  X;B\right)  _{\sigma}+I\left(
Y;B|X\right)  _{\sigma}. \label{eq:entropy-1}%
\end{equation}
An accessible introduction to concepts in quantum Shannon theory is available
in Yard's thesis~\cite{Yard05a}.

\section{The private dynamic capacity theorem}

\label{sec:capacity-theorem}The private dynamic capacity theorem gives bounds
on the reliable communication rates of a noisy quantum channel when combined
with the noiseless resources of public classical communication, private
classical communication, and a secret key. The theorem applies regardless of
whether a protocol consumes the noiseless resources or generates them.

\begin{theorem}
[Private Dynamic Capacity]\label{thm:main-theorem}The private dynamic capacity
region $\mathcal{C}_{\text{\emph{RPS}}}(\mathcal{N})$ of a quantum channel
$\mathcal{N}$ is equal to the following expression:%
\begin{equation}
\mathcal{C}_{\text{\emph{RPS}}}(\mathcal{N})=\overline{\bigcup_{k=1}^{\infty
}\frac{1}{k}\mathcal{C}_{\text{\emph{RPS}}}^{(1)}(\mathcal{N}^{\otimes k}%
)},\label{eq:multi-letter}%
\end{equation}
where the overbar indicates the closure of a set. The \textquotedblleft
one-shot\textquotedblright\ region $\mathcal{C}_{\text{\emph{RPS}}}%
^{(1)}(\mathcal{N})$ is the union of the \textquotedblleft one-shot,
one-state\textquotedblright\ regions $\mathcal{C}_{\text{\emph{RPS}},\sigma
}^{(1)}(\mathcal{N})$:%
\[
\mathcal{C}_{\text{\emph{RPS}}}^{(1)}(\mathcal{N})\equiv\bigcup_{\sigma
}\mathcal{C}_{\text{\emph{RPS}},\sigma}^{(1)}(\mathcal{N}).
\]
The \textquotedblleft one-shot, one-state\textquotedblright\ region
$\mathcal{C}_{\text{\emph{RPS}},\sigma}^{(1)}(\mathcal{N})$ is the set of all
rates $R$, $P$, and $S$ such that%
\begin{align}
R+P &  \leq I\left(  YX;B\right)  _{\sigma},\label{eq:RP-bound}\\
P+S &  \leq I\left(  Y;B|X\right)  _{\sigma}-I\left(  Y;E|X\right)  _{\sigma
},\label{eq:PS-bound}\\
R+P+S &  \leq I\left(  YX;B\right)  _{\sigma}-I\left(  Y;E|X\right)  _{\sigma
}.\label{eq:RPS-bound}%
\end{align}
The above entropic quantities are with respect to a classical-quantum state
$\sigma^{XYBE}$ where%
\begin{equation}
\sigma^{XYBE}\equiv\sum_{x,y}p_{X,Y}\left(  x,y\right)  \left\vert
x\right\rangle \left\langle x\right\vert ^{X}\otimes\left\vert y\right\rangle
\left\langle y\right\vert ^{Y}\otimes U_{\mathcal{N}}^{A^{\prime}\rightarrow
BE}(\rho_{x,y}^{A^{\prime}}),\label{eq:main-theorem-state}%
\end{equation}
and the states $\rho_{x,y}^{A^{\prime}}$ are mixed. It is implicit that one
should consider states on $A^{\prime k}$ instead of $A^{\prime}$ when taking
the regularization in (\ref{eq:multi-letter}).
\end{theorem}

The above theorem is a \textquotedblleft multi-letter\textquotedblright%
\ capacity theorem because of the regularization in (\ref{eq:multi-letter}).
Though, we show later that the regularization is not necessary for
entanglement-breaking channels, the Hadamard class of channels, or the quantum
erasure channels. We prove the private dynamic capacity theorem in two parts:

\begin{enumerate}
\item The direct coding theorem below shows that combining the
\textquotedblleft publicly-enhanced private father protocol\textquotedblright%
\ with the one-time pad, secret key distribution, and private-to-public
transmission achieves the above region.

\item The converse theorem demonstrates that any coding scheme cannot do
better than the regularization in (\ref{eq:multi-letter}), in the sense that a
scheme with vanishing error should have its rates below the above amounts. We
prove the converse theorem directly in \textquotedblleft one fell
swoop,\textquotedblright\ by employing a catalytic, information-theoretic
approach (similar to the method introduced in Ref.~\cite{WH10}).
\end{enumerate}

\section{Dynamic achievable rate region}

\label{sec:achievable}The unit resource achievable region is what Alice and
Bob can achieve with the protocols secret key distribution, the one-time pad,
and private-to-public transmission. It is the cone of the rate triples
corresponding to these protocols:%
\[
\left\{  \alpha\left(  0,-1,1\right)  +\beta\left(  -1,1,-1\right)
+\gamma\left(  1,-1,0\right)  :\alpha,\beta,\gamma\geq0\right\}  .
\]
We can also write any rate triple $\left(  R,P,S\right)  $ in the unit
resource capacity region with a matrix equation:%
\begin{equation}%
\begin{bmatrix}
R\\
P\\
S
\end{bmatrix}
=%
\begin{bmatrix}
0 & -1 & 1\\
-1 & 1 & -1\\
1 & -1 & 0
\end{bmatrix}%
\begin{bmatrix}
\alpha\\
\beta\\
\gamma
\end{bmatrix}
. \label{eq:unit-resource-achievable-region}%
\end{equation}
The inverse of the above matrix is as follows:%
\[%
\begin{bmatrix}
-1 & -1 & 0\\
-1 & -1 & -1\\
0 & -1 & -1
\end{bmatrix}
,
\]
and gives the following set of inequalities for the unit resource achievable
region:%
\begin{align*}
R+P  &  \leq0,\\
R+P+S  &  \leq0,\\
P+S  &  \leq0,
\end{align*}
by inverting the matrix equation in (\ref{eq:unit-resource-achievable-region})
and applying the constraints $\alpha,\beta,\gamma\geq0$.

Now, let us include the publicly-enhanced private father protocol~\cite{HW09}.
Ref.~\cite{HW09} proved that we can achieve the following rate triple by
channel coding over a noisy quantum channel $\mathcal{N}^{A^{\prime
}\rightarrow B}$:%
\[
\left(  I\left(  X;B\right)  _{\sigma},I\left(  Y;B|X\right)  _{\sigma
},-I\left(  Y;E|X\right)  _{\sigma}\right)  ,
\]
for any state $\sigma^{XYBE}$\ of the form:%
\begin{equation}
\sigma^{XYBE}\equiv\sum_{x,y}p_{X,Y}\left(  x,y\right)  \left\vert
x\right\rangle \left\langle x\right\vert ^{X}\otimes\left\vert y\right\rangle
\left\langle y\right\vert ^{Y}\otimes U_{\mathcal{N}}^{A^{\prime}\rightarrow
BE}(\rho_{x,y}^{A^{\prime}}), \label{eq:maximization-state}%
\end{equation}
where $U_{\mathcal{N}}^{A^{\prime}\rightarrow BE}$ is an isometric extension
of the quantum channel $\mathcal{N}^{A^{\prime}\rightarrow B}$. Specifically,
we showed in Ref.~\cite{HW09}\ that one can achieve the above rates with
vanishing error in the limit of large blocklength. Thus the achievable rate
region is the following translation of the unit resource achievable region in
(\ref{eq:unit-resource-achievable-region}):%
\[%
\begin{bmatrix}
R\\
P\\
S
\end{bmatrix}
=%
\begin{bmatrix}
0 & -1 & 1\\
-1 & 1 & -1\\
1 & -1 & 0
\end{bmatrix}%
\begin{bmatrix}
\alpha\\
\beta\\
\gamma
\end{bmatrix}
+%
\begin{bmatrix}
I\left(  X;B\right)  _{\sigma}\\
I\left(  Y;B|X\right)  _{\sigma}\\
-I\left(  Y;E|X\right)  _{\sigma}%
\end{bmatrix}
.
\]
We can now determine bounds on an achievable rate region that employs the
above coding strategy. We apply the inverse of the matrix in
(\ref{eq:unit-resource-achievable-region}) to the LHS and RHS. Then using the
constraints $\alpha,\beta,\gamma\geq0$, we obtain the inequalities in
(\ref{eq:RP-bound}-\ref{eq:RPS-bound}), corresponding exactly to the one-shot,
one-state region in Theorem~\ref{thm:main-theorem}. Taking the union over all
possible states $\sigma$ in (\ref{eq:maximization-state}) and taking the
regularization gives the full private dynamic achievable rate region.

\section{Catalytic and information theoretic converse proof}

\label{sec:converse}This section provides a catalytic, information theoretic
converse proof of the private dynamic capacity region, showing that
(\ref{eq:multi-letter}) gives a multi-letter characterization of it. The
catalytic approach means that we are considering the most general protocol
that \textit{consumes and generates} public classical communication, private
classical communication, and secret key in addition to the uses of the noisy
quantum channel. Figure~\ref{fig:catalytic-secret-protocol} depicts the most
general protocol for generating public classical communication, private
classical communication, and a secret key with the consumption of a noisy
quantum channel $\mathcal{N}^{A^{\prime}\rightarrow B}$ and the same
respective resources. This approach has the advantage that we can prove the
converse theorem in \textquotedblleft one fell swoop.\textquotedblright\ We
employ the Alicki-Fannes' inequality, the chain rule for quantum mutual
information, elementary properties of quantum entropy, and the quantum data
processing inequality to prove the converse.

There are some subtleties in our proof for the converse theorem. We prove that
the bounds in (\ref{eq:RP-bound}-\ref{eq:RPS-bound}) hold for common
randomness generation and private key generation instead of public classical
communication and private classical communication, respectively, because a
capacity for generating common randomness and a private key can only be better
than that for generating public classical communication and private classical
communication. This setting is slightly different from that depicted in
Figure~\ref{fig:catalytic-secret-protocol}.%

\begin{figure}
[ptb]
\begin{center}
\includegraphics[
natheight=4.893100in,
natwidth=7.886200in,
height=3.141in,
width=5.047in
]%
{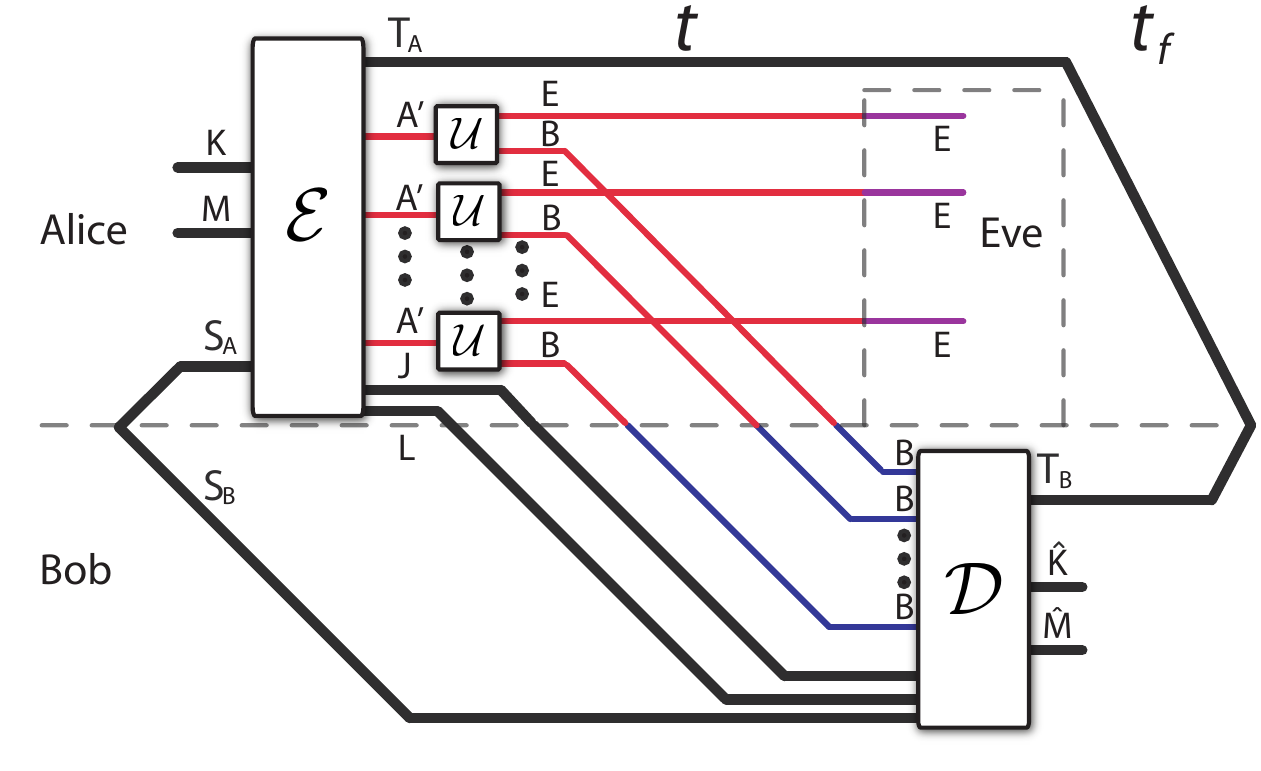}%
\caption{(Color online) The most general protocol for generating public
classical communication, private classical communication, and secret key with
the help of the same respective resources and many uses of a noisy quantum
channel. Alice begins with her public classical register $K$, her private
classical register $M$, and her share of the secret key in $S_{A}$. She
encodes according to some CPTP\ encoding map $\mathcal{E}$ that outputs a
classical register $T_{A}$, many quantum registers $A^{\prime n}$, a public
classical register $L$, and a private classical register $J$. She inputs
$A^{\prime n}$ to many uses of the noisy channel $\mathcal{N}^{A^{\prime
}\rightarrow B}$ (with isometric extension $U_{\mathcal{N}}^{A^{\prime
}\rightarrow BE}$), transmits $J$ over a noiseless private classical channel,
and transmits $L$ over a noiseless public classical channel. Bob receives the
channel outputs $B^{n}$, the private classical register $J$, and the public
classical register $L$ and performs a decoding $\mathcal{D}$ that recovers the
public and private classical information. The decoding also generates secret
key with system $T_{A}$.}%
\label{fig:catalytic-secret-protocol}%
\end{center}
\end{figure}

We prove that the converse theorem holds for a state of the following form:%
\begin{equation}
\sigma^{XYBE}\equiv\sum_{x,y}p_{X,Y}(x,y)\left\vert x\right\rangle
\left\langle x\right\vert ^{X}\otimes\left\vert y\right\rangle \left\langle
y\right\vert ^{Y}\otimes U_{\mathcal{N}}^{A^{\prime}\rightarrow BE}(\rho
_{x,y}^{A^{\prime}}), \label{eq:converse-state}%
\end{equation}
where the states $\rho_{x,y}^{A^{\prime}}$ are mixed.

We prove all three bounds in (\ref{eq:RP-bound}-\ref{eq:RPS-bound}). Alice
possesses the following classical registers:

\begin{enumerate}
\item Two public classical registers $K$ and $K_{A}$\ in the maximally
correlated state $\overline{\Phi}^{KK_{A}}$ where the dimension of both
systems is $2^{n\overline{R}}$. The register $K_{A}$ is for public classical communication.

\item Two private classical registers $M$ and $M_{A}$ in the maximally
correlated state $\overline{\Phi}^{MM_{A}}$ \ where the dimension of both
systems is $2^{n\overline{P}}$. The register $M_{A}$ is for private classical communication.

\item One share $S_{A}$\ of a secret key. The shared secret key is in the
maximally correlated state $\overline{\Phi}^{S_{A}S_{B}}$ where the dimension
of both systems is $2^{n\widetilde{S}}$. Bob possesses the other share $S_{B}$
of the secret key.
\end{enumerate}

Our convention above is that the protocol generates a resource whose rate has
an overbar and consumes a resource whose rate has a tilde.

The initial state is as follows:%
\[
\omega^{MM_{A}KK_{A}S_{A}S_{B}}\equiv\overline{\Phi}^{MM_{A}}\otimes
\overline{\Phi}^{KK_{A}}\otimes\overline{\Phi}^{S_{A}S_{B}}.
\]
She passes the registers $K_{A}$, $M_{A}$, and $S_{A}$\ into an encoding map
$\mathcal{E}^{K_{A}M_{A}S_{A}\rightarrow A^{\prime n}T_{A}LJ}$. This map
outputs a classical register $T_{A}$ of dimension $2^{n\overline{S}}$, a
public classical register $L$ of dimension $2^{n\widetilde{R}}$, a private
classical register $J$ of dimension $2^{n\widetilde{P}}$, and many quantum
systems $A^{\prime n}$ for input to the channel. The register $T_{A}$\ is for
creating a secret key with Bob. The state after this encoding map is as
follows:%
\[
\omega^{MKS_{B}A^{\prime n}T_{A}LJ}\equiv\mathcal{E}^{K_{A}M_{A}%
S_{A}\rightarrow A^{\prime n}T_{A}LJ}(\omega^{MM_{A}KK_{A}S_{A}S_{B}}).
\]
She sends the systems $A^{\prime n}$ through many uses $\mathcal{N}^{A^{\prime
n}\rightarrow B^{n}}$ of the noisy channel $\mathcal{N}^{A^{\prime}\rightarrow
B}$, transmits $L$ over a noiseless public classical channel, and transmits
$J$ over a noiseless private classical channel, producing the following state:%
\begin{equation}
\omega^{MKB^{n}E^{n}T_{A}LJS_{B}}\equiv U_{\mathcal{N}}^{A^{\prime
n}\rightarrow B^{n}E^{n}}(\omega^{MKS_{B}A^{\prime n}T_{A}LJ}),
\label{eq:post-channel-state}%
\end{equation}
where $U_{\mathcal{N}}^{A^{\prime n}\rightarrow B^{n}E^{n}}$ is the isometric
extension of the channel $\mathcal{N}^{A^{\prime n}\rightarrow B^{n}}$. The
above state is a state of the form in (\ref{eq:converse-state}) with $X\equiv
KL$ and $Y\equiv MJS_{B}T_{A}$. Bob then applies a map $\mathcal{D}%
^{B^{n}S_{B}LJ\rightarrow T_{B}\hat{M}\hat{K}}$ that outputs classical
registers $T_{B},\hat{M},\hat{K}$. Let $\omega^{\prime}$ denote the final state.

The following condition should hold for a catalytic private dynamic protocol
that transmits the public and private classical data and establishes secret
key with $\epsilon$-error:%
\begin{equation}
\left\Vert \overline{\Phi}^{M\hat{M}}\otimes\rho^{K\hat{K}E^{n}}%
\otimes\overline{\Phi}^{T_{A}T_{B}}-\left(  \omega^{\prime}\right)  ^{M\hat
{M}K\hat{K}E^{n}T_{A}T_{B}}\right\Vert _{1}\leq\epsilon
,\label{eq:+++_good-code}%
\end{equation}
where $\rho^{K\hat{K}E^{n}}$ is some state such that Tr$_{E^{n}}\{\rho
^{K\hat{K}E^{n}}\}=\overline{\Phi}^{K\hat{K}}$. Condition
(\ref{eq:+++_good-code}) implies that Alice and Bob establish maximal
classical correlations in $M$ and $\hat{M}$, in $K$ and $\hat{K}$, and in
$T_{A}$ and $T_{B}$. The following security condition should hold as well:%
\[
\left\Vert \omega^{MKE^{n}T_{A}LJS_{B}}-\pi^{MT_{A}JS_{B}}\otimes
\sigma^{KLE^{n}}\right\Vert _{1}\leq\epsilon,
\]
where $\omega^{MKE^{n}T_{A}LJS_{B}}$ is the state
in\ (\ref{eq:post-channel-state}) obtained from tracing over Bob's systems,
$\pi$ is the maximally mixed state, and $\sigma^{KLE^{n}}$ is some state on
the public registers and Eve's systems. This security criterion implies that
Eve cannot learn anything about any of the private data if she has access to
all of the public data in addition to her registers. It also implies that the
following information-theoretic bound holds:%
\begin{equation}
I\left(  MJS_{B}T_{A};E^{n}KL\right)  _{\omega}\leq\epsilon
.\label{eq_converse_cond2}%
\end{equation}
The net rate triple for the protocol is as follows: $(\overline{R}%
-\widetilde{R},\overline{P}-\widetilde{P},\overline{S}-\widetilde{S})$. The
protocol generates a resource if its corresponding rate is positive, and it
consumes a resource if its corresponding rate is negative.

We prove the first bound in (\ref{eq:RP-bound}). Consider the following chain
of inequalities:%
\begin{align*}
n\left(  \overline{R}+\overline{P}\right)   &  =I(KM;\hat{K}\hat
{M})_{\overline{\Phi}\otimes\overline{\Phi}}\\
&  \leq I(KM;\hat{K}\hat{M})_{\omega^{\prime}}+n\delta^{\prime}\\
&  \leq I\left(  KM;B^{n}LJS_{B}\right)  _{\omega}\\
&  =I(KM;B^{n}LJ|S_{B})_{\omega}\\
&  =H(KMS_{B})_{\omega}+H(B^{n}LJS_{B})_{\omega}-H(KMB^{n}LJS_{B})_{\omega
}-H(S_{B})_{\omega}\\
&  \leq H(KMS_{B})_{\omega}+H(B^{n})_{\omega}+H(LJS_{B})_{\omega}%
-H(KMB^{n}LJS_{B})_{\omega}-H(S_{B})_{\omega}\\
&  =I(KLMJS_{B};B^{n})_{\omega}-H(KLMJS_{B})_{\omega}+H(KMS_{B})_{\omega
}+H(LJS_{B})_{\omega}-H(S_{B})_{\omega}\\
&  =I(KLMJS_{B};B^{n})_{\omega}+I(KM;LJ|S_{B})_{\omega}\\
&  \leq I(KLMJS_{B}T_{A};B^{n})_{\omega}+I(KM;LJ|S_{B})_{\omega}\\
&  \leq I(XY;B^{n})_{\omega}+n(\widetilde{R}+\widetilde{P}).
\end{align*}
The first equality follows by evaluating the mutual information $I(MK;\hat
{M}\hat{K})$ on the state $\overline{\Phi}^{K\hat{K}}\otimes\overline{\Phi
}^{M\hat{M}}$. The first inequality follows from the condition in
(\ref{eq:+++_good-code}) and an application of the Alicki-Fannes' inequality
where $\delta^{\prime}$ vanishes as $\epsilon\rightarrow0$. We suppress this
term in the rest of the inequalities for convenience. The second inequality
follows from quantum data processing. The second equality follows by applying
the mutual information chain rule in (\ref{eq:mut-chain-rule}) and because
$I(KM;S_{B})_{\omega}=0$ for this protocol. The third equality follows from
expanding the conditional mutual information $I(KM;B^{n}LJ|S_{B})_{\omega}$.
The third inequality follows by subadditivity of the entropy $H(B^{n}%
LJS_{B})_{\omega}$. The fourth equality follows because%
\[
H(B^{n})_{\omega}-H(KMB^{n}LJS_{B})_{\omega}=I(KLMJS_{B};B^{n})_{\omega
}-H(KLMJS_{B})_{\omega},
\]
and the fifth equality follows because%
\[
-H(KLMJS_{B})_{\omega}+H(KMS_{B})_{\omega}+H(LJS_{B})_{\omega}-H(S_{B}%
)_{\omega}=I(KM;LJ|S_{B})_{\omega}.
\]
The fourth inequality follows from quantum data processing. The final
inequality follows from the definitions $X\equiv KL$ and $Y\equiv MJS_{B}%
T_{A}$ and because the quantum mutual information $I(KM;LJ|S_{B})_{\omega}$
can never be larger than the logarithm of the dimension of the classical
registers $LJ$.

We now prove the bound in (\ref{eq:PS-bound}). Consider the following chain of
inequalities:%
\begin{align*}
n\left(  \overline{P}+\overline{S}\right)   &  =I(MT_{A};\hat{M}%
T_{B})_{\overline{\Phi}\otimes\overline{\Phi}}\\
&  \leq I(MT_{A};\hat{M}T_{B})_{\omega^{\prime}}+n\delta^{\prime}\\
&  \leq I(MT_{A};B^{n}JLKS_{B})_{\omega}\\
&  \leq I(MT_{A};B^{n}JLKS_{B})_{\omega}-I(MT_{A}JS_{B};E^{n}KL)_{\omega
}+\epsilon\\
&  =I(MT_{A};B^{n}JS_{B}|KL)_{\omega}+I(MT_{A};KL)_{\omega}-I(MT_{A}%
JS_{B};E^{n}|KL)_{\omega}-I(MT_{A}JS_{B};KL)_{\omega}+\epsilon\\
&  =I(MT_{A}JS_{B};B^{n}|KL)_{\omega}+I(MT_{A};JS_{B}|KL)_{\omega}%
-I(B^{n};JS_{B}|KL)_{\omega}\\
&  \ \ \ \ \ \ \ +I(MT_{A};KL)_{\omega}-I(MT_{A}JS_{B};KL)_{\omega}%
-I(MT_{A}JS_{B};E^{n}|KL)_{\omega}+\epsilon\\
&  \leq I(MT_{A}JS_{B};B^{n}|KL)_{\omega}-I(MT_{A}JS_{B};E^{n}|KL)_{\omega
}+I(MT_{A};JS_{B}|KL)_{\omega}+\epsilon\\
&  \leq I(Y;B^{n}|X)_{\omega}-I(Y;E^{n}|X)_{\omega}+n(\widetilde{P}%
+\widetilde{S})+\epsilon.
\end{align*}
The first equality follows by evaluating the entropy for the state
$\overline{\Phi}^{T_{A}T_{B}}\otimes\overline{\Phi}^{M\hat{M}}$. The first
inequality follows from the condition in (\ref{eq:+++_good-code}) and an
application of the Alicki-Fannes' inequality where $\delta^{\prime}%
\rightarrow0$ as $\epsilon\rightarrow0$. We suppress this term in the rest of
the inequalities for convenience. The second inequality follows from quantum
data processing. The third inequality follows from the bound in
(\ref{eq_converse_cond2}) on Eve's information. The second and third
equalities follow from the chain rule for quantum mutual information. The
fourth inequality follows from quantum data processing $I(MT_{A}%
JS_{B};KL)_{\omega}\geq I(MT_{A};KL)_{\omega}$ and the fact that
$I(B^{n};JS_{B}|KL)_{\omega}\geq0$. The last inequality follows from the
definitions $X\equiv KL$ and $Y\equiv MJS_{B}T_{A}$ and because the mutual
information $I(MT_{A};JS_{B}|KL)_{\omega}$ can never be larger than the
logarithm of the dimensions of the registers $J,S_{B}$.

We finally prove the bound in (\ref{eq:RPS-bound}). Consider the following
chain of inequalities:%
\begin{align*}
n\left(  \overline{R}+\overline{P}+\overline{S}\right)   &  =I(KMT_{A};\hat
{K}\hat{M}T_{B})_{\overline{\Phi}\otimes\overline{\Phi}\otimes\overline{\Phi}%
}\\
&  \leq I(KMT_{A};\hat{K}\hat{M}T_{B})_{\omega^{\prime}}+n\delta^{\prime}\\
&  \leq I(KMT_{A};B^{n}JLS_{B})_{\omega}\\
&  \leq I(KMT_{A};B^{n}JLS_{B})_{\omega}-I(MT_{A}JS_{B};E^{n}|KL)_{\omega
}+\epsilon\\
&  =I(KLMT_{A}JS_{B};B^{n})_{\omega}+I\left(  JLS_{B};KMT_{A}\right)
_{\omega}-I\left(  JLS_{B};B^{n}\right)  _{\omega}\\
&  \ \ \ \ \ \ \ -I(MT_{A}JS_{B};E^{n}|KL)_{\omega}+\epsilon\\
&  \leq I(YX;B^{n})_{\omega}-I(Y;E^{n}|X)_{\omega}+n(\widetilde{R}%
+\widetilde{P}+\widetilde{S})+\epsilon.
\end{align*}
The first equality follows by evaluating the entropy for the state
$\overline{\Phi}^{K\hat{K}}\otimes\overline{\Phi}^{M\hat{M}}\otimes
\overline{\Phi}^{T_{A}T_{B}}$. The first inequality follows from the condition
in (\ref{eq:+++_good-code}) and an application of the Alicki-Fannes'
inequality where $\delta^{\prime}\rightarrow0$ as $\epsilon\rightarrow0$. We
suppress this term in the rest of the inequalities for convenience. The second
inequality follows from quantum data processing. The third inequality follows
from the condition in (\ref{eq_converse_cond2}) (note that
$I(MT_{A}JS_{B};E^{n}|KL)_{\omega} + I(MT_{A}JS_{B};KL)_{\omega} = 
I(MT_{A}JS_{B};E^{n}KL)_{\omega}$ from the chain rule and both terms on the
LHS are non-negative). The second equality follows
from the chain rule for quantum mutual information. The final inequality
follows from the definitions $X\equiv KL$ and $Y\equiv MJS_{B}T_{A}$, because
$I\left(  JLS_{B};B^{n}\right)  _{\omega}\geq0$, and because the mutual
information $I\left(  JLS_{B};KMT_{A}\right)  _{\omega}$ can never be larger
than the logarithm of the dimensions of the registers $J,L,S_{B}$.

\section{The private dynamic capacity formula}

\label{sec:dynamic-cap-formula}The private dynamic capacity formula is a
particular formula that is relevant in the computation of the private dynamic
capacity region. If this formula is additive for a particular channel, then
the computation of the region simplifies, in the sense that it requires an optimization
over a single use of the channel, rather than with an infinite number of
them~\cite{BV04}. The reasoning for this is similar to our discussion in
Section~6 of Ref.~\cite{WH10}, appealing to ideas from Pareto-optimal
trade-off analysis (see Chapter~4\ of Ref.~\cite{BV04}). Thus, we keep the
discussion to a minimum here and instead refer the reader to Section~6 of
Ref.~\cite{WH10} for further explanations.

\begin{definition}
[Private Dynamic Capacity Formula]The private dynamic capacity formula of a
quantum channel $\mathcal{N}$ is as follows:%
\begin{equation}
P_{\lambda,\mu}\left(  \mathcal{N}\right)  \equiv\max_{\sigma}I\left(
YX;B\right)  _{\sigma}+\lambda\left[  I\left(  Y;B|X\right)  _{\sigma
}-I\left(  Y;E|X\right)  _{\sigma}\right]  +\mu\left[  I\left(  YX;B\right)
_{\sigma}-I\left(  Y;E|X\right)  _{\sigma}\right]  , \label{eq:objective}%
\end{equation}
where $\lambda,\mu\geq0$.
\end{definition}

\begin{definition}
The regularized private dynamic capacity formula is as follows:%
\[
P_{\lambda,\mu}^{\text{reg}}\left(  \mathcal{N}\right)  \equiv\lim
_{n\rightarrow\infty}\frac{1}{n}P_{\lambda,\mu}\left(  \mathcal{N}^{\otimes
n}\right)  .
\]
\end{definition}

\begin{lemma}
\label{thm:CEQ-single-letter}Suppose the private dynamic capacity formula is
additive for channels $\mathcal{N}$ and $\mathcal{M}$:%
\[
P_{\lambda,\mu}\left(  \mathcal{N\otimes M}\right)  =P_{\lambda,\mu}\left(
\mathcal{N}\right) +P_{\lambda,\mu}\left(
\mathcal{M}\right) .
\]
Then the regularized private dynamic capacity formula is equal to the private
dynamic capacity formula:%
\[
P_{\lambda,\mu}^{\text{reg}}\left(  \mathcal{N}\right)  =P_{\lambda,\mu
}\left(  \mathcal{N}\right)  .
\]
In this sense, the regularized formula \textquotedblleft
single-letterizes\textquotedblright\ and it is not necessary to take the limit.
\end{lemma}

\begin{proof}
The proof is similar to the proof of Lemma~1 in Ref.~\cite{WH10}.
\end{proof}

\begin{theorem}
Single-letterization of the private dynamic capacity formula implies that the
computation of the Pareto optimal trade-off surface of the private dynamic
capacity region requires an optimization over a single channel use.
\end{theorem}

\begin{proof}
The proof exploits the same techniques as the proof of Theorem~2 in
Ref.~\cite{WH10}.
\end{proof}

\subsection{Special cases of the private dynamic capacity formula}

We now consider several special cases of the private dynamic capacity formula.
These special cases have similar geometric interpretations as discussed in
Section~6.1 of Ref.~\cite{WH10}. The first case corresponds to considering a
supporting hyperplane of the capacity region with normal vector $\left(
1,1,0\right)  $, the second corresponds to considering a supporting hyperplane
with normal vector $\left(  0,1,1\right)  $, and the last a supporting
hyperplane with normal vector $\left(  1,1,1\right)  $. Each of these choices
corresponds to singling out only one of the inequalities in
Theorem~\ref{thm:main-theorem}\ and maximizing with respect to that inequality.

\begin{corollary}
The private dynamic capacity formula is equivalent to the HSW\ classical
capacity formula~\cite{Hol98,SW97}\ when $\lambda,\mu=0$, in the sense that%
\[
\max_{\sigma}I\left(  YX;B\right)  _{\sigma}=\max_{\rho^{XA^{\prime}}}I\left(
X;B\right)  _{\rho},
\]
where%
\[
\rho^{XA^{\prime}}\equiv\sum_{x}p_{X}\left(  x\right)  \left\vert
x\right\rangle \left\langle x\right\vert ^{X}\otimes\rho_{x}^{A^{\prime}},
\]
and $\sigma$ is a state of the form in Theorem~\ref{thm:main-theorem}.
\end{corollary}

\begin{proof}
The proof of this statement follows merely by redefining the joint classical
variable $XY$ in the first formula to be the classical variable $X$ in the
second formula.
\end{proof}

\begin{corollary}
The private dynamic capacity formula is equivalent to the
Devetak-Cai-Winter-Yeung private classical capacity formula
\cite{Devetak03,CWY04}\ in the limit where $\lambda\rightarrow\infty$ and
$\mu$ is fixed, in the sense that%
\[
\max_{\sigma}\left[  I\left(  Y;B|X\right)  _{\sigma}-I\left(  Y;E|X\right)
_{\sigma}\right]  =\max_{\rho^{XA^{\prime}}}\left[  I\left(  X;B\right)
_{\rho}-I\left(  X;E\right)  _{\rho}\right]  ,
\]
where $\rho$ is a state of the form in the above corollary and $\sigma$ is a
state of the form in Theorem~\ref{thm:main-theorem}.
\end{corollary}

\begin{proof}
The inequality LHS $\geq$ RHS follows by choosing the distribution
$p_{X,Y}\left(  x,y\right)  =p_{X}\left(  x\right)  p_{Y|X}\left(  y|x\right)
$ with $p_{Y|X}\left(  y|x\right)  =p_{X}^{\ast}\left(  y\right)  $ and
$p_{X}\left(  x\right)  =\delta_{x,x_{0}}$ and choosing the conditional
density operators $\rho_{x_{0},y}^{A^{\prime}}=\left(  \rho_{x}^{\ast}\right)
^{A^{\prime}}$ where the asterisked quantities are optimal for the RHS. The
inequality LHS $\leq$ RHS follows because the quantity $I\left(  Y;B|X\right)
_{\sigma}-I\left(  Y;E|X\right)  _{\sigma}=\sum_{x}p_{X}\left(  x\right)
\left[  I\left(  Y;B\right)  _{\sigma_{x}}-I\left(  Y;E\right)  _{\sigma_{x}%
}\right]  $ and an average is always less than a maximum.
\end{proof}

\begin{corollary}
The private dynamic capacity formula is equivalent to the HSW\ classical
capacity formula in the limit where $\mu\rightarrow\infty$ and $\lambda$ is
fixed, in the sense that%
\[
\max_{\sigma}\left[  I\left(  YX;B\right)  _{\sigma}-I\left(  Y;E|X\right)
_{\sigma}\right]  =\max_{\left\{  p_{X}\left(  x\right)  ,\psi_{x}\right\}
}I\left(  X;B\right)  .
\]
\end{corollary}

\begin{proof}
The inequality LHS $\geq$ RHS follows by choosing the distribution
$p_{X,Y}\left(  x,y\right)  =p_{X}\left(  x\right)  p_{Y|X}\left(  y|x\right)
$ with $p_{Y|X}\left(  y|x\right)  =\delta_{y,y_{0}}$ and $p_{X}\left(
x\right)  =p_{X}^{\ast}\left(  x\right)  $ and choosing the conditional
density operators $\rho_{x,y_{0}}^{A^{\prime}}=\left(  \rho_{x}^{\ast}\right)
^{A^{\prime}}$. The inequality LHS $\leq$ RHS follows because $$
I\left(  YX;B\right)  _{\sigma}-I\left(  Y;E|X\right)  _{\sigma}  \leq
I\left(  YX;B\right)  _{\sigma}\leq\max_{\left\{  p_{X}\left(  x\right)
,\psi_{x}\right\}  }I\left(  X;B\right) . $$
\end{proof}

\subsubsection{Comparison between public-private and classical-quantum
regions}

We now compare the Devetak-Shor classical-quantum trade-off
formula~\cite{DS03} with a special case of our above formula that applies to a
trade-off between public and private classical communication. We should expect
these two formulas to be comparable from the Collins-Popescu analogy because
no entanglement or secret key is involved. The result is that the
public-private region is generally larger than the classical-quantum region,
but the two regions are equivalent for degradable quantum channels.

First consider the following refinement of the Devetak-Shor\ formula (see
Section~IV-A-4 of Ref.~\cite{HW08ITIT}):%
\[
f_{\mu}\left(  \mathcal{N}\right)  \equiv\max_{\rho}I\left(  X;B\right)
_{\rho}+I\left(  A\rangle BX\right)  _{\rho}+\mu I\left(  A\rangle BX\right)
_{\rho},
\]
where $\rho^{XAB}$ is a state of the form%
\[
\rho^{XAB}\equiv\sum_{x}p_{X}\left(  x\right)  \left\vert x\right\rangle
\left\langle x\right\vert ^{X}\otimes\mathcal{N}^{A^{\prime}\rightarrow
B}(\phi_{x}^{AA^{\prime}}).
\]
The formula for the public-private trade-off is a special case of the private
dynamic capacity formula:%
\[
P_{\mu}\left(  \mathcal{N}\right)  \equiv\max_{\sigma}I\left(  X;B\right)
_{\sigma}+I\left(  Y;B|X\right)  _{\sigma}-I\left(  Y;E|X\right)  _{\sigma
}+\mu\left[  I\left(  Y;B|X\right)  _{\sigma}-I\left(  Y;E|X\right)  _{\sigma
}\right]  ,
\]
where $\sigma^{XYBE}$ is a state of the form%
\[
\sigma^{XYBE}\equiv\sum_{x,y}p_{X,Y}\left(  x,y\right)  \left\vert
x\right\rangle \left\langle x\right\vert ^{X}\otimes\left\vert y\right\rangle
\left\langle y\right\vert ^{Y}\otimes U_{\mathcal{N}}^{A^{\prime}\rightarrow
BE}(\rho_{x,y}^{A^{\prime}}),
\]

\begin{lemma}
The classical-quantum trade-off formula is never greater than the
public-private trade-off formula:%
\[
f_{\mu}\left(  \mathcal{N}\right)  \leq P_{\mu}\left(  \mathcal{N}\right)  .
\]
\end{lemma}

\begin{proof}
The proof techniqiue is similar to that of Lemma~3 in
Ref.~\cite{PhysRevA.78.022306}. First let us rewrite the function $f$ so that
it is a function on the systems $X$, $B$, and $E$:%
\begin{align*}
f_{\mu}\left(  \mathcal{N}\right)   &  =\max_{\rho}I\left(  X;B\right)
_{\rho}+I\left(  A\rangle BX\right)  _{\rho}+\mu I\left(  A\rangle BX\right)
_{\rho}\\
&  =\max_{\rho}I\left(  X;B\right)  _{\rho}+\left(  \mu+1\right)  \left[
H\left(  B|X\right)  _{\rho}-H\left(  E|X\right)  _{\rho}\right]  .
\end{align*}
Thus, it is only important to consider the input system $A^{\prime}$ when
evaluating the above formula. Let
\[
\rho_{x}^{A^{\prime}}=\text{Tr}_{A}\{\phi_{x}^{AA^{\prime}}\},
\]
so that the maximization above is over a state of the following form:%
\[
\rho^{XBE}\equiv\sum_{x}p_{X}\left(  x\right)  \left\vert x\right\rangle
\left\langle x\right\vert ^{X}\otimes U_{\mathcal{N}}^{A^{\prime}\rightarrow
BE}(\rho_{x}^{A^{\prime}}),
\]
where $U$ is the isometric extension of the channel $\mathcal{N}$. Take a
spectral decomposition of the states $\rho_{x}^{A^{\prime}}$:%
\[
\rho_{x}^{A^{\prime}}=\sum_{y}p_{Y|X}\left(  y|x\right)  \psi_{x,y}%
^{A^{\prime}},
\]
where the states $\psi_{x,y}^{A^{\prime}}$ are pure. Then the following state
$\theta^{XYBE}$ is a particular state of the form $\sigma^{XYBE}$:%
\[
\theta^{XYBE}\equiv\sum_{x,y}p_{X}\left(  x\right)  p_{Y|X}\left(  y|x\right)
\left\vert x\right\rangle \left\langle x\right\vert ^{X}\otimes\left\vert
y\right\rangle \left\langle y\right\vert ^{Y}\otimes U_{\mathcal{N}%
}^{A^{\prime}\rightarrow BE}(\psi_{x,y}^{A^{\prime}}),
\]
such that Tr$_{Y}\left\{  \theta\right\}  =\rho^{XBE}$.
Consider the following chain of inequalities:%
\begin{align*}
&  I\left(  X;B\right)  _{\rho}+\left(  \mu+1\right)  \left[  H\left(
B|X\right)  _{\rho}-H\left(  E|X\right)  _{\rho}\right]  \\
&  =I\left(  X;B\right)  _{\theta}+\left(  \mu+1\right)  \left[  H\left(
B|X\right)  _{\theta}-H\left(  E|X\right)  _{\theta}\right]  \\
&  =I\left(  X;B\right)  _{\theta}+\left(  \mu+1\right)  \left[  H\left(
B|X\right)  _{\theta}-H\left(  B|YX\right)  _{\theta}-H\left(  E|X\right)
_{\theta}+H\left(  E|YX\right)  _{\theta}\right]  \\
&  =I\left(  X;B\right)  _{\theta}+\left(  \mu+1\right)  \left[  I\left(
Y;B|X\right)  _{\theta}-I\left(  Y;E|X\right)  _{\theta}\right]  \\
&  \leq P_{\mu}\left(  \mathcal{N}\right)  .
\end{align*}
The first equality follows because Tr$_{Y}\left\{  \theta\right\}  =\rho
^{XBE}$. The second equality follows because the entropies of $\theta$ on
systems $B$ and $E$ are equal when conditioned on $X$ and $Y$. The third
equality follows from the definition of conditional mutual information. The
final inequality follows from the definition of $P_{\mu}\left(  \mathcal{N}%
\right)  $.
\end{proof}

\begin{lemma}
Suppose that a quantum channel is degradable. Then the classical-quantum
trade-off formula is equivalent to the public-private trade-off formula.
\end{lemma}

\begin{proof}
The proof is again similar to that of Lemma~3 in
Ref.~\cite{PhysRevA.78.022306}. Consider the state definitions in the previous
lemma and the definition of $\sigma^{XYBE}$ from before. Consider a state
$\sigma^{XYZBE}$ defined as follows:%
\[
\sigma^{XYZBE}\equiv\sum_{x,y,z}p_{X,Y}\left(  x,y\right)  p_{Z|X,Y}\left(
z|x,y\right)  \left\vert x\right\rangle \left\langle x\right\vert ^{X}%
\otimes\left\vert y\right\rangle \left\langle y\right\vert ^{Y}\otimes
\left\vert z\right\rangle \left\langle z\right\vert ^{Z}\otimes U_{\mathcal{N}%
}^{A^{\prime}\rightarrow BE}(\varphi_{x,y,z}^{A^{\prime}}),
\]
where $\rho_{x,y}^{A^{\prime}}=\sum_{z}p_{Z|X,Y}\left(  z|x,y\right)
\varphi_{x,y,z}^{A^{\prime}}$ is a spectral decomposition of $\rho
_{x,y}^{A^{\prime}}$. Consider the following chain of inequalities that
applies to an arbitrary state $\sigma^{XYBE}$:%
\begin{align*}
&  I\left(  X;B\right)  _{\sigma}+\left(  \mu+1\right)  \left[  I\left(
Y;B|X\right)  _{\sigma}-I\left(  Y;E|X\right)  _{\sigma}\right]  \\
&  =I\left(  X;B\right)  _{\sigma}+\left(  \mu+1\right)  \left[  I\left(
YZ;B|X\right)  _{\sigma}-I\left(  Z;B|XY\right)  _{\sigma}-\left[  I\left(
YZ;E|X\right)  _{\sigma}-I\left(  Z;E|XY\right)  _{\sigma}\right]  \right]  \\
&  =I\left(  X;B\right)  _{\sigma}+\left(  \mu+1\right)  \left[  I\left(
YZ;B|X\right)  _{\sigma}-I\left(  YZ;E|X\right)  _{\sigma}-\left[  I\left(
Z;B|XY\right)  _{\sigma}-I\left(  Z;E|XY\right)  _{\sigma}\right]  \right]  \\
&  \leq I\left(  X;B\right)  _{\sigma}+\left(  \mu+1\right)  \left[  I\left(
YZ;B|X\right)  _{\sigma}-I\left(  YZ;E|X\right)  _{\sigma}\right]  \\
&  =I\left(  X;B\right)  _{\sigma}+\left(  \mu+1\right)  \left[  H\left(
B|X\right)  _{\sigma}-H\left(  B|XYZ\right)  _{\sigma}-H\left(  E|X\right)
_{\sigma}+H\left(  E|XYZ\right)  _{\sigma}\right]  \\
&  =I\left(  X;B\right)  _{\sigma}+\left(  \mu+1\right)  \left[  H\left(
B|X\right)  _{\sigma}-H\left(  E|X\right)  _{\sigma}\right]  \\
&  \leq f_{\mu}\left(  \mathcal{N}\right)  .
\end{align*}
The first equality follows by applying the chain rule for quantum mutual
information. The second equality follows by rearranging terms. The first
inequality follows because $I\left(  Z;B|XY\right)  _{\sigma}-I\left(
Z;E|XY\right)  _{\sigma}\geq0$ for a degradable quantum channel. The third
equality follows by expanding mutual informations. The fourth equality follows
because the entropies of the state $\sigma$ on systems $B$ and $E$ are equal
when conditioned on $X$, $Y$, and $Z$. The final inequality follows from the
definition of $f_{\mu}\left(  \mathcal{N}\right)  $.
\end{proof}

\subsubsection{Comparison between quantum dynamic and private dynamic
formulas}

We can compare the quantum dynamic and private dynamic capacity formulas for
the class of degradable channels. The proof exploits the simplified form of
the private dynamic capacity formula that results from
Lemma~\ref{lem:degradable-formula}, and the quantum dynamic capacity formula
appears in the proof below.

\begin{lemma}
Suppose that a quantum channel $\mathcal{N}_{\emph{D}}$ is degradable. Then
the quantum dynamic capacity formula can never be less than the private
dynamic capacity formula.
\end{lemma}

\begin{proof}
We prove this theorem by showing that%
\[
P_{\lambda,\mu}(\mathcal{N}_{\text{D}})\leq D_{\lambda,\mu}\left(
\mathcal{N}_{D}\right)  ,
\]
where $D_{\lambda,\mu}(\mathcal{N}_{\text{D}})$ is the quantum dynamic
capacity formula given by%
\[
D_{\lambda,\mu}(\mathcal{N}_{\text{D}})\equiv\max_{\sigma}I\left(
AX;B\right)  _{\sigma}+\lambda I\left(  A\rangle BX\right)  _{\sigma}%
+\mu\left[  I\left(  X;B\right)  _{\sigma}+I\left(  A\rangle BX\right)
_{\sigma}\right]  .
\]
The state $\sigma^{XABE}$ is a state of the form%
\[
\sigma^{XABE}\equiv\sum_{x}p_{X}\left(  x\right)  \left\vert x\right\rangle
\left\langle x\right\vert ^{X}\otimes U_{\mathcal{N}_{\text{D}}}^{A^{\prime
}\rightarrow BE}(\phi_{x}^{AA^{\prime}}),
\]
where the states $\phi_{x}^{AA^{\prime}}$ are pure.
Suppose that the following state is the one that maximizes $P_{\lambda,\mu
}(\mathcal{N}_{\text{D}})$ for a given $\lambda$ and $\mu$:%
\[
\omega^{XYBE}\equiv\sum_{x,y}p_{X}\left(  x\right)  p_{Y|X}\left(  y|x\right)
\left\vert x\right\rangle \left\langle x\right\vert ^{X}\otimes\left\vert
y\right\rangle \left\langle y\right\vert ^{Y}\otimes U_{\mathcal{N}_{\text{D}%
}}^{A^{\prime}\rightarrow BE}(\psi_{x,y}^{A^{\prime}})
\]
Then we choose the states $\phi_{x}^{AA^{\prime}}$ in a given $\sigma^{XABE}$
to be as follows:%
\[
\left\vert \phi_{x}\right\rangle ^{AA^{\prime}}=\sum_{y}\sqrt{p_{Y|X}\left(
y|x\right)  }\left\vert y\right\rangle ^{A}\left\vert \psi_{x,y}\right\rangle
^{A^{\prime}}.
\]
We can obtain the state $\omega^{XYBE}$ from the state $\sigma^{XABE}$ by
performing a complete dephasing $\Delta^{A\rightarrow Y}$ where the dephasing
basis is $\left\{  \left\vert y\right\rangle \left\langle y\right\vert
\right\}  $. Then the following inequalities hold%
\begin{align*}
&  I\left(  YX;B\right)  _{\omega}+\lambda\left[  H\left(  B|X\right)
_{\omega}-H\left(  E|X\right)  _{\omega}\right]  +\mu\left[  H\left(
B\right)  _{\omega}-H\left(  E|X\right)  _{\omega}\right]  \\
&  =I\left(  YX;B\right)  _{\omega}+\lambda\left[  H\left(  B|X\right)
_{\sigma}-H\left(  E|X\right)  _{\sigma}\right]  +\mu\left[  H\left(
B\right)  _{\sigma}-H\left(  E|X\right)  _{\sigma}\right]  \\
&  \leq I\left(  AX;B\right)  _{\sigma}+\lambda\left[  H\left(  B|X\right)
_{\sigma}-H\left(  E|X\right)  _{\sigma}\right]  +\mu\left[  H\left(
B\right)  _{\sigma}-H\left(  E|X\right)  _{\sigma}\right]  \\
&  =I\left(  AX;B\right)  _{\sigma}+\lambda I\left(  A\rangle BX\right)
_{\sigma}+\mu\left[  I\left(  X;B\right)  _{\sigma}+I\left(  A\rangle
BX\right)  _{\sigma}\right]  \\
&  \leq D_{\lambda,\mu}(\mathcal{N}_{\text{D}}).
\end{align*}
The first equality follows because the entropies of $\omega$ and $\sigma$
without the $Y$ system are equivalent. The first inequality follows from
quantum data processing:\ one can obtain the state $\omega^{XYBE}$ be
performing a von Neumann measurement of the $A$ system of the state
$\sigma^{XABE}$ in the basis $\{\left\vert y\right\rangle ^{A}\}$. The second
equality follows by rearranging terms. The final equality follows from the
definition of $D_{\lambda,\mu}$.
\end{proof}

The above lemma explicitly shows how the analogy between the classical and
quantum worlds breaks down for the case of a degradable channel. The quantum
dynamic capacity formula is always larger than the private dynamic formula
because of the strong correlations in entanglement and because of the lack of
a super-dense coding protocol in the public-private-secret-key setting.

\section{Single-letter private dynamic capacity regions for
entanglement-breaking channels}

\label{sec:EB}Our first class of channels for which the private dynamic
capacity region simplifies is the class of entanglement-breaking channels.
Shor found that such channels have an additive classical capacity
\cite{shor:4334}, and we can extend his method of proof to show that the full
private dynamic capacity region for these channels is single-letter.

\begin{theorem}
[Private Dynamic Capacity for Entanglement-Breaking Channels]%
\label{thm:EB-region}The private dynamic capacity region $\mathcal{C}%
_{\text{\emph{RPS}}}(\mathcal{N}_{\emph{EB}})$ of an entanglement-breaking
quantum channel $\mathcal{N}_{\emph{EB}}$ is the set of all rates $R$, $P$,
and $S$, such that%
\begin{align}
R+P  &  \leq\max_{\omega}I\left(  X;B\right)  _{\omega},\\
P+S  &  \leq0,\\
R+P+S  &  \leq\max_{\omega}I\left(  X;B\right)  _{\omega}.
\end{align}
The above entropic quantities are with respect to a classical-quantum state
$\sigma^{XB}$ where%
\begin{equation}
\sigma^{XB}\equiv\sum_{x}p_{X}\left(  x\right)  \left\vert x\right\rangle
\left\langle x\right\vert ^{X}\otimes\mathcal{N}_{\emph{EB}}^{A^{\prime
}\rightarrow B}(\psi_{x}^{A^{\prime}}),
\end{equation}
and the states $\psi_{x}^{A^{\prime}}$ are pure.
\end{theorem}

We prove this theorem in a few steps. We first show that the private dynamic
capacity formula simplifies dramatically for antidegradable channels (recall
that entanglement-breaking channels are a special case of antidegradable
ones). We then show that this simplified formula is additive for an
entanglement-breaking channel and this result implies the form of the region
in the statement of the above theorem.

\begin{lemma}
\label{lem:antidegradable-formula}Suppose that a quantum channel
$\mathcal{N}_{\text{\emph{AD}}}$ is antidegradable. Then the private dynamic
capacity formula simplifies as follows:%
\[
P_{\lambda,\mu}\left(  \mathcal{N}_{\text{\emph{AD}}}\right)  =h_{\lambda,\mu
}\left(  \mathcal{N}_{\text{\emph{AD}}}\right)  ,
\]
where%
\[
h_{\lambda,\mu}\left(  \mathcal{N}_{\text{\emph{AD}}}\right)  \equiv\left(
1+\mu\right)  \max_{\omega}I\left(  X;B\right)  _{\omega},
\]
and $\omega^{XBE}$ is a state of the following form:%
\[
\omega^{XBE}\equiv\sum_{x}p_{X}\left(  x\right)  \left\vert x\right\rangle
\left\langle x\right\vert ^{X}\otimes U_{\mathcal{N}_{\text{\emph{AD}}}%
}^{A^{\prime}\rightarrow BE}(\psi_{x}^{A^{\prime}}),
\]
and the states $\psi_{x}^{A^{\prime}}$ are pure.
\end{lemma}

\begin{proof}
The inequality $P_{\lambda,\mu}(\mathcal{N}_{\text{AD}})\geq h_{\lambda,\mu
}(\mathcal{N}_{\text{AD}})$ follows by carefully choosing the state
$\sigma^{XYBE}$ for the maximization on the LHS: choose the distribution
$p_{X,Y}\left(  x,y\right)  =p_{X}^{\ast}\left(  x\right)  \delta_{y,y_{0}}$
and each state $\rho_{x,y_{0}}^{A^{\prime}}=\left(  \psi_{x}^{\ast}\right)
^{A^{\prime}}$ where the terms with asterisks are optimal for the RHS.
The other inequality $P_{\lambda,\mu}(\mathcal{N}_{\text{AD}})\leq
h_{\lambda,\mu}(\mathcal{N}_{\text{AD}})$ follows from the following chain of
inequalities:%
\begin{align*}
&  I\left(  YX;B\right)  _{\sigma}+\lambda\left[  I\left(  Y;B|X\right)
_{\sigma}-I\left(  Y;E|X\right)  _{\sigma}\right]  +\mu\left[  I\left(
YX;B\right)  _{\sigma}-I\left(  Y;E|X\right)  _{\sigma}\right]  \\
&  \leq I\left(  YX;B\right)  _{\sigma}+\mu I\left(  YX;B\right)  _{\sigma}\\
&  \leq h_{\lambda,\mu}(\mathcal{N}_{\text{AD}}).
\end{align*}
The first inequality follows because $\left[  I\left(  Y;B|X\right)  _{\sigma
}-I\left(  Y;E|X\right)  _{\sigma}\right]  \leq0$ (from antidegradability) and
by dropping the term $-\mu I\left(  Y;E|X\right)  _{\sigma}$. The second
inequality follows because $I\left(  YX;B\right)  _{\sigma}\leq\max_{\omega
}I\left(  X;B\right)  _{\omega}$.
\end{proof}

\begin{corollary}
The private dynamic capacity formula is additive for an antidegradable channel
$\mathcal{N}_{\emph{AD}}$\ and an entanglement-breaking channel $\mathcal{N}%
_{\emph{EB}}$:%
\[
P_{\lambda,\mu}(\mathcal{N}_{\emph{AD}}\otimes\mathcal{N}_{\emph{EB}%
})=P_{\lambda,\mu}(\mathcal{N}_{\emph{AD}})+P_{\lambda,\mu}(\mathcal{N}%
_{\emph{EB}})
\]
\end{corollary}

\begin{proof}
The proof is similar to the proof in Ref.~\cite{shor:4334}. We first note that
the tensor product of an antidegradable channel and an entanglement-breaking
channel is an antidegradable channel. This observation allows us to employ the
simplified formula in Lemma~\ref{lem:antidegradable-formula}. We employ the
following states in the proof:%
\begin{align*}
\omega^{XB} &  \equiv\sum_{x}p_{X}\left(  x\right)  \left\vert x\right\rangle
\left\langle x\right\vert ^{X}\otimes\mathcal{N}_{\text{AD}}^{A_{1}\rightarrow
B_{1}}\otimes\mathcal{N}_{\text{EB}}^{A_{2}\rightarrow B_{2}}(\psi_{x}%
^{A_{1}A_{2}})\\
&  =\sum_{x}p_{X}\left(  x\right)  \left\vert x\right\rangle \left\langle
x\right\vert ^{X}\otimes\sum_{z}p_{Z|X}\left(  z|x\right)  \mathcal{N}%
_{\text{AD}}^{A_{1}\rightarrow B_{1}}(\sigma_{z,x}^{A_{1}})\otimes\theta
_{z,x}^{B_{2}},\\
\omega^{XZB} &  \equiv\sum_{x,z}p_{X}\left(  x\right)  p_{Z|X}\left(
z|x\right)  \left\vert x\right\rangle \left\langle x\right\vert ^{X}%
\otimes\left\vert z\right\rangle \left\langle z\right\vert ^{Z}\otimes
\mathcal{N}_{\text{AD}}^{A_{1}\rightarrow B_{1}}(\sigma_{z,x}^{A_{1}}%
)\otimes\theta_{z,x}^{B_{2}}.
\end{align*}
Consider the following chain of inequalities:%
\begin{align*}
P_{\lambda,\mu}(\mathcal{N}_{\text{AD}}\otimes\mathcal{N}_{\text{EB}}) &
=h_{\lambda,\mu}(\mathcal{N}_{\text{AD}}\otimes\mathcal{N}_{\text{EB}})\\
&  =\left(  1+\mu\right)  I\left(  X;B_{1}B_{2}\right)  _{\omega}\\
&  =\left(  1+\mu\right)  \left[  H\left(  B_{1}B_{2}\right)  _{\omega
}-H\left(  B_{1}B_{2}|X\right)  _{\omega}\right]  \\
&  \leq\left(  1+\mu\right)  \left[  H\left(  B_{1}\right)  _{\omega}+H\left(
B_{2}\right)  _{\omega}-H\left(  B_{1}B_{2}|XZ\right)  _{\omega}\right]  \\
&  =\left(  1+\mu\right)  \left[  H\left(  B_{1}\right)  _{\omega}+H\left(
B_{2}\right)  _{\omega}-H\left(  B_{1}|XZ\right)  _{\omega}-H\left(
B_{2}|XZ\right)  _{\omega}\right]  \\
&  =\left(  1+\mu\right)  \left[  I\left(  X;B_{1}\right)  _{\omega}+I\left(
X;B_{2}\right)  _{\omega}\right]  \\
&  \leq h_{\lambda,\mu}(\mathcal{N}_{\text{AD}})+h_{\lambda,\mu}%
(\mathcal{N}_{\text{EB}})\\
&  =P_{\lambda,\mu}(\mathcal{N}_{\text{AD}})+P_{\lambda,\mu}(\mathcal{N}%
_{\text{EB}}).
\end{align*}
The first equality follows from Lemma~\ref{lem:antidegradable-formula}. The
second equality follows from the assumption that $\omega$ is a state that
maximizes $h_{\lambda,\mu}(\mathcal{N}_{\text{AD}}\otimes\mathcal{N}%
_{\text{EB}})$. The third equality follows from the definition of quantum
mutual information. The first inequality follows from subadditivity of entropy
and conditioning does not increase entropy. The fourth equality follows
because the state $\omega$ is product when conditioned on both $X$ and $Z$.
The fifth equality follows from the definition of quantum mutual information.
The second inequality follows because the mutual informations are always less
than their maxima, and the final equality follows from
Lemma~\ref{lem:antidegradable-formula}.
\end{proof}

\begin{example}
The private dynamic capacity region of a completely dephasing channel is the
set of all $R$, $P$, and $S$ satisfying the following inequalities:%
\begin{align*}
R+P  &  \leq1,\\
P+S  &  \leq0,\\
R+P+S  &  \leq1.
\end{align*}
This result follows because the completely dephasing channel is an
entanglement-breaking channel with public classical capacity equal to one.
\end{example}

\section{Single-letter private dynamic capacity regions for the quantum
Hadamard channels}

\label{sec:Hadamard}We now prove that the private dynamic capacity region is
additive for the class of quantum Hadamard channels. This result is perhaps
dual to the above result because Hadamard channels are ones for which the map
to the environment is entanglement-breaking, and they are degradable with a
degrading map from Bob to the environment Eve. Our method of proof is similar
as above---we first prove that the private dynamic capacity formula simplifies
for degradable channels and then prove additivity of the simplified formula
for the Hadamard channels.

\begin{lemma}
\label{lem:degradable-formula}Suppose that a quantum channel $\mathcal{N}%
_{\emph{D}}$ is degradable. Then the private dynamic capacity formula
simplifies as follows:%
\[
P_{\lambda,\mu}(\mathcal{N}_{\emph{D}})=g_{\lambda,\mu}(\mathcal{N}_{\emph{D}%
}),
\]
where%
\[
g_{\lambda,\mu}(\mathcal{N}_{\emph{D}})\equiv\max_{\omega}I\left(
YX;B\right)  _{\omega}+\lambda\left[  H\left(  B|X\right)  _{\omega}-H\left(
E|X\right)  _{\omega}\right]  +\mu\left[  H\left(  B\right)  _{\omega
}-H\left(  E|X\right)  _{\omega}\right]  ,
\]
and $\omega^{XYBE}$ is a state of the following form:%
\[
\omega^{XYBE}\equiv\sum_{x,y}p_{X,Y}\left(  x,y\right)  \left\vert
x\right\rangle \left\langle x\right\vert ^{X}\otimes\left\vert y\right\rangle
\left\langle y\right\vert ^{Y}\otimes U_{\mathcal{N}_{\emph{D}}}^{A^{\prime
}\rightarrow BE}(\psi_{x,y}^{A^{\prime}}),
\]
and the states $\psi_{x,y}^{A^{\prime}}$ are pure.
\end{lemma}

\begin{proof}
The inequality $P_{\lambda,\mu}(\mathcal{N}_{\text{D}})\geq g_{\lambda,\mu
}(\mathcal{N}_{\text{D}})$ follows by choosing each state $\rho_{x,y}%
^{A^{\prime}}$ in $\sigma^{XYBE}$ for the maximization on the LHS\ to be the
pure state $\psi_{x,y}^{A^{\prime}}$ that maximizes the RHS. Consider the
following chain of inequalities:%
\begin{align*}
P_{\lambda,\mu}(\mathcal{N}_{\text{D}})  &  \geq I\left(  YX;B\right)
_{\sigma}+\lambda\left[  I\left(  Y;B|X\right)  _{\sigma}-I\left(
Y;E|X\right)  _{\sigma}\right]  +\mu\left[  I\left(  YX;B\right)  _{\sigma
}-I\left(  Y;E|X\right)  _{\sigma}\right] \\
&  =I\left(  YX;B\right)  _{\sigma}+\lambda\left[  H\left(  B|X\right)
_{\sigma}-H\left(  B|XY\right)  _{\sigma}-H\left(  E|X\right)  _{\sigma
}+H\left(  E|XY\right)  _{\sigma}\right] \\
&  \ \ \ \ \ \ +\mu\left[  H\left(  B\right)  _{\sigma}-H\left(  B|XY\right)
_{\sigma}-H\left(  E|X\right)  _{\sigma}+H\left(  E|XY\right)  _{\sigma
}\right] \\
&  =I\left(  YX;B\right)  _{\sigma}+\lambda\left[  H\left(  B|X\right)
_{\sigma}-H\left(  E|X\right)  _{\sigma}\right]  +\mu\left[  H\left(
B\right)  _{\sigma}-H\left(  E|X\right)  _{\sigma}\right] \\
&  =g_{\lambda,\mu}(\mathcal{N}_{\text{D}}).
\end{align*}
We now prove that the other inequality $P_{\lambda,\mu}(\mathcal{N}_{\text{D}%
})\leq g_{\lambda,\mu}(\mathcal{N}_{\text{D}})$\ holds. Suppose the state
$\sigma^{XYBE}$ maximizes $P_{\lambda,\mu}(\mathcal{N}_{\text{D}})$. Consider
a state $\sigma^{XYZBE}$ defined as follows:%
\[
\sigma^{XYZBE}\equiv\sum_{x,y,z}p_{X,Y}\left(  x,y\right)  p_{Z|X,Y}\left(
z|x,y\right)  \left\vert x\right\rangle \left\langle x\right\vert ^{X}%
\otimes\left\vert y\right\rangle \left\langle y\right\vert ^{Y}\otimes
\left\vert z\right\rangle \left\langle z\right\vert ^{Z}\otimes U_{\mathcal{N}%
_{\text{D}}}^{A^{\prime}\rightarrow BE}(\varphi_{x,y,z}^{A^{\prime}}),
\]
where $\rho_{x,y}^{A^{\prime}}=\sum_{z}p_{Z|X,Y}\left(  z|x,y\right)
\varphi_{x,y,z}^{A^{\prime}}$ is a spectral decomposition of each $\rho
_{x,y}^{A^{\prime}}$ in the state $\sigma^{XYBE}$. This state is a state of
the form $\omega^{XYBE}$ with $Y$ redefined to be $YZ$. Consider the following
chain of inequalities:%
\begin{align*}
P_{\lambda,\mu}(\mathcal{N}_{\text{D}})  &  =I\left(  YX;B\right)  _{\sigma
}+\lambda\left[  I\left(  Y;B|X\right)  _{\sigma}-I\left(  Y;E|X\right)
_{\sigma}\right]  +\mu\left[  I\left(  YX;B\right)  _{\sigma}-I\left(
Y;E|X\right)  _{\sigma}\right] \\
&  =I\left(  YZX;B\right)  _{\sigma}+\lambda\left[  I\left(  YZ;B|X\right)
_{\sigma}-I\left(  YZ;E|X\right)  _{\sigma}-\left[  I\left(  Z;B|YX\right)
_{\sigma}-I\left(  Z;E|YX\right)  _{\sigma}\right]  \right] \\
&  \ \ \ \ \ \ +\mu\left[  I\left(  X;B\right)  _{\sigma}+I\left(
YZ;B|X\right)  _{\sigma}-I\left(  YZ;E|X\right)  _{\sigma}-\left[  I\left(
Z;B|YX\right)  _{\sigma}-I\left(  Z;E|YX\right)  _{\sigma}\right]  \right] \\
&  \leq I\left(  YZX;B\right)  _{\sigma}+\lambda\left[  I\left(
YZ;B|X\right)  _{\sigma}-I\left(  YZ;E|X\right)  _{\sigma}\right] \\
&  \ \ \ \ \ \ +\mu\left[  I\left(  X;B\right)  _{\sigma}+I\left(
YZ;B|X\right)  _{\sigma}-I\left(  YZ;E|X\right)  _{\sigma}\right] \\
&  =I\left(  YZX;B\right)  _{\sigma}+\lambda\left[  H\left(  B|X\right)
_{\sigma}-H\left(  B|XYZ\right)  _{\sigma}-H\left(  E|X\right)  _{\sigma
}+H\left(  E|XYZ\right)  _{\sigma}\right] \\
&  \ \ \ \ \ \ +\mu\left[  H\left(  B\right)  _{\sigma}-H\left(  B|XYZ\right)
_{\sigma}-H\left(  E|X\right)  _{\sigma}+H\left(  E|XYZ\right)  _{\sigma
}\right] \\
&  =I\left(  YZX;B\right)  _{\sigma}+\lambda\left[  H\left(  B|X\right)
_{\sigma}-H\left(  E|X\right)  _{\sigma}\right]  +\mu\left[  H\left(
B\right)  _{\sigma}-H\left(  E|X\right)  _{\sigma}\right] \\
&  \leq g_{\lambda,\mu}(\mathcal{N}_{\text{D}}).
\end{align*}
The first equality follows by definition. The second equality follows from
applying the chain rule for mutual information. The first inequality follows
because $I\left(  Z;B|YX\right)  _{\sigma}-I\left(  Z;E|YX\right)  _{\sigma
}\geq0$ for a degradable channel. The third equality follows by expanding the
mutual informations, and the fourth equality follows because $H\left(
B|XYZ\right)  _{\sigma}=H\left(  E|XYZ\right)  _{\sigma}$. The final
inequality follows from the definition of $g_{\lambda,\mu}(\mathcal{N}%
_{\text{D}})$.
\end{proof}

\begin{lemma}
Suppose that $\mathcal{N}_{\emph{H}}$ is a quantum Hadamard channel and that
$\mathcal{N}_{\emph{D}}$ is a degradable quantum channel. Then the private
dynamic capacity formula is additive:%
\[
P_{\lambda,\mu}(\mathcal{N}_{\emph{H}}\otimes\mathcal{N}_{\emph{D}%
})=P_{\lambda,\mu}(\mathcal{N}_{\emph{H}})+P_{\lambda,\mu}(\mathcal{N}%
_{\emph{D}}).
\]
\end{lemma}

\begin{proof}
The inequality $P_{\lambda,\mu}(\mathcal{N}_{\text{H}}\otimes\mathcal{N}%
_{\text{D}})\geq P_{\lambda,\mu}(\mathcal{N}_{\text{H}})+P_{\lambda,\mu
}(\mathcal{N}_{\text{D}})$ trivially holds by picking the state on the LHS\ to
be a tensor product of the ones that individually maximize the RHS.
Thus, we prove the non-trivial inequality $P_{\lambda,\mu}(\mathcal{N}%
_{\text{H}}\otimes\mathcal{N}_{\text{D}})\leq P_{\lambda,\mu}(\mathcal{N}%
_{\text{H}})+P_{\lambda,\mu}(\mathcal{N}_{\text{D}})$ for the channels in the
hypothesis of the lemma. Consider a state of the form $\sigma^{XYB_{1}%
E_{1}B_{2}E_{2}}$ that arises from inputting a state of the form in
Lemma~\ref{lem:degradable-formula} to the tensor product channel. Let
$\omega^{XYZWE_{1}B_{2}E_{2}}$ be the state that arises from applying the
first part of the degrading map of the Hadamard channel to system $B_{1}$.
Then the following chain of inequalities holds:%
\begin{align*}
&  P_{\lambda,\mu}(\mathcal{N}_{\text{H}}\otimes\mathcal{N}_{\text{D}})\\
&  =g_{\lambda,\mu}(\mathcal{N}_{\text{H}}\otimes\mathcal{N}_{\text{D}})\\
&  =H\left(  B_{1}B_{2}\right)  _{\sigma}-H\left(  E_{1}E_{2}|YX\right)
_{\sigma}+\lambda\left[  H\left(  B_{1}B_{2}|X\right)  _{\sigma}-H\left(
E_{1}E_{2}|X\right)  _{\sigma}\right]  +\mu\left[  H\left(  B_{1}B_{2}\right)
_{\sigma}-H\left(  E_{1}E_{2}|X\right)  _{\sigma}\right] \\
&  =H\left(  B_{1}\right)  _{\sigma}-H\left(  E_{1}|YX\right)  _{\sigma
}+\lambda\left[  H\left(  B_{1}|X\right)  _{\sigma}-H\left(  E_{1}|X\right)
_{\sigma}\right]  +\mu\left[  H\left(  B_{1}\right)  _{\sigma}-H\left(
E_{1}|X\right)  _{\sigma}\right] \\
&  \ \ \ \ \ \ +H\left(  B_{2}|B_{1}\right)  _{\sigma}-H\left(  E_{2}%
|YXE_{1}\right)  _{\sigma}+\lambda\left[  H\left(  B_{2}|XB_{1}\right)
_{\sigma}-H\left(  E_{2}|XE_{1}\right)  _{\sigma}\right]  +\mu\left[  H\left(
B_{2}|B_{1}\right)  _{\sigma}-H\left(  E_{2}|XE_{1}\right)  _{\sigma}\right]
\\
&  \leq H\left(  B_{1}\right)  _{\sigma}-H\left(  E_{1}|YX\right)  _{\sigma
}+\lambda\left[  H\left(  B_{1}|X\right)  _{\sigma}-H\left(  E_{1}|X\right)
_{\sigma}\right]  +\mu\left[  H\left(  B_{1}\right)  _{\sigma}-H\left(
E_{1}|X\right)  _{\sigma}\right] \\
&  \ \ \ \ \ \ +H\left(  B_{2}\right)  _{\sigma}-H\left(  E_{2}|YXW\right)
_{\sigma}+\lambda\left[  H\left(  B_{2}|XW\right)  _{\sigma}-H\left(
E_{2}|XW\right)  _{\sigma}\right]  +\mu\left[  H\left(  B_{2}\right)
_{\sigma}-H\left(  E_{2}|XW\right)  _{\sigma}\right] \\
&  \leq g_{\lambda,\mu}(\mathcal{N}_{\text{H}})+g_{\lambda,\mu}(\mathcal{N}%
_{\text{D}})\\
&  =P_{\lambda,\mu}(\mathcal{N}_{\text{H}})+P_{\lambda,\mu}(\mathcal{N}%
_{\text{D}}).
\end{align*}
The first equality follows from Lemma~\ref{lem:degradable-formula} because a
Hadamard channel is degradable and thus the tensor product channel is
degradable as well. The second equality follows by definition. The third
equality follows by expanding with the chain rule for entropy. The first
inequality follows from subadditivity ($H\left(  B_{2}|B_{1}\right)  _{\sigma
}\leq H\left(  B_{2}\right)  _{\sigma}$) and because there is a degrading map
from $B_{1}\rightarrow W$ and from $W\rightarrow E_{1}$ (and so $H\left(
B_{2}|XB_{1}\right)  _{\sigma}\leq H\left(  B_{2}|XW\right)  _{\sigma}$ and
$H\left(  E_{2}|XW\right)  _{\sigma}\leq H\left(  E_{2}|XE_{1}\right)
_{\sigma}$). The second inequality follows from the definition of
$g_{\lambda,\mu}$, and the final equality follows from
Lemma~\ref{lem:degradable-formula}.
\end{proof}

\section{The private dynamic capacity region for special channels}

\label{sec:special-channels}In the forthcoming subsections, we explicitly
compute and plot the private dynamic capacity region for the qubit dephasing
channel, the $1\rightarrow N$ cloning channel, and the quantum erasure
channel. Interestingly, the ensemble required to achieve the boundary is the
same for all three boundaries. The proofs of the theorems in this section are
similar (though with subtle differences) to proofs from
Refs.~\cite{BHTW10,WH10}, and they all appear in the appendix.

\subsection{Dephasing channels}

\label{sec:dephasing}Consider the qubit dephasing channel $\mathcal{N}_{p}$
with dephasing probability $p$:
\begin{equation}
\mathcal{N}_{p}(\rho):=(1-p)\rho+p\Delta(\rho),
\end{equation}
where $\Delta(\rho):=\langle0|\rho|0\rangle|0\rangle\!\langle0|+\langle
1|\rho|1\rangle|1\rangle\!\langle1|$ is the completely dephasing channel. The
below theorem gives an explicit form for the private dynamic capacity region
of this channel, and Figure~\ref{fig:dephasing} plots the region for a
dephasing parameter $p=0.2$.%
\begin{figure}
[ptb]
\begin{center}
\includegraphics[
natheight=3.893400in,
natwidth=7.580100in,
height=2.6057in,
width=5.047in
]%
{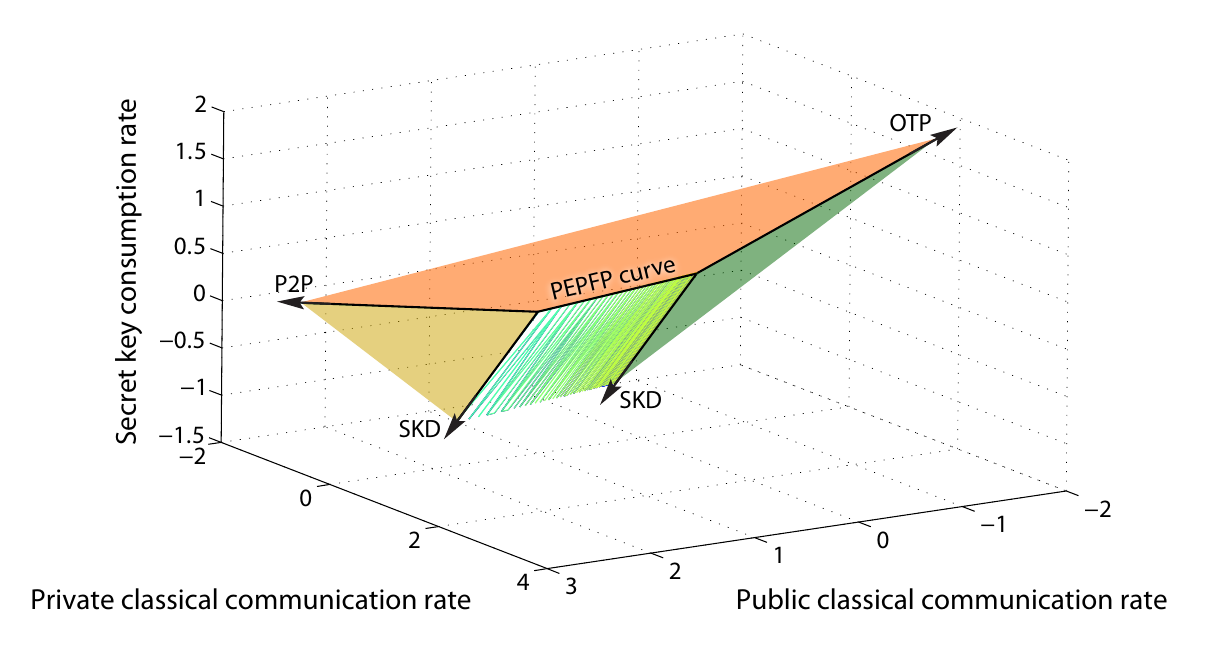}%
\caption{(Color online) The private dynamic triple trade-off for the qubit
dephasing channel with dephasing parameter $p=0.2$. P2P is in the direction of
private-to-public transmission, SKD is in the direction of secret key
distribution, OTP is in the direction of the one-time pad protocol, and PEPFP
is the publicly-enhanced private father trade-off curve (this convention is
the same in the forthcoming figures). The region exhibits a non-trivial
resource trade-off only on the surface below the PEPFP\ trade-off curve in the
direction of secret key distribution.}%
\label{fig:dephasing}%
\end{center}
\end{figure}

\begin{theorem}
\label{thm:dephasing}The private dynamic capacity region $\mathcal{C}%
_{\mathrm{{RPS}}}(\mathcal{N}_{p})$ of a dephasing channel with dephasing
parameter $p$\ is the set of all $R$, $P$, and $S$ such that%
\begin{align}
R+P  &  \leq1,\\
P+S  &  \leq H_{2}\left(  \nu\right)  -H_{2}(\gamma\left(  \nu,p\right)  ),\\
R+P+S  &  \leq1-H_{2}(\gamma\left(  \nu,p\right)  ),
\end{align}
where $\nu\in\left[  0,1/2\right]  $, $H_{2}$ is the binary entropy function,
and%
\[
\gamma\left(  \nu,p\right)  \equiv\frac{1}{2}+\frac{1}{2}\sqrt{1-16\cdot
\frac{p}{2}\left(  1-\frac{p}{2}\right)  \nu(1-\nu)}.
\]
\end{theorem}

\subsection{Quantum cloning channels}

A $1\rightarrow N$ cloning channel \cite{BHP09,B09,BDHM09,BHTW10} is the map
induced by a universal cloning machine~\cite{GM97}. It approximately copies
the input with a maximal fidelity independent of the input. The communication
model for this channel gives all of the approximate clones to the receiver Bob
and gives the environment of the map to Eve. The Kraus operators for a
$1\rightarrow N$ cloning channel are as follows:%
\[
\left\{  \frac{1}{\sqrt{\Delta_{N}}}\left(  \sqrt{N-i}\left\vert
i\right\rangle ^{B}\left\langle 0\right\vert ^{A^{\prime}}+\sqrt
{i+1}\left\vert i+1\right\rangle ^{B}\left\langle 1\right\vert ^{A^{\prime}%
}\right)  \right\}  _{i=0}^{N-1},
\]
where $\Delta_{N}\equiv N\left(  N+1\right)  /2$ and%
\[
\{\left\vert j\right\rangle ^{B}\equiv\left\vert N-j,j\right\rangle
\}_{j=0}^{N},
\]
where $\left\vert N-j,j\right\rangle ^{B}$ denotes a normalized state on an
$N$-qubit system that is a uniform superposition of computational basis states
with $N-j$ \textquotedblleft zeros\textquotedblright\ and $j$
\textquotedblleft ones.\textquotedblright\ Figure~\ref{fig:cloning-RPS}\ plots
the capacity region for a $1\rightarrow10$ cloning channel, and the proof of
the below theorem appears in the appendix.%
\begin{figure}
[ptb]
\begin{center}
\includegraphics[
natheight=5.013300in,
natwidth=7.959700in,
height=3.5042in,
width=5.5486in
]%
{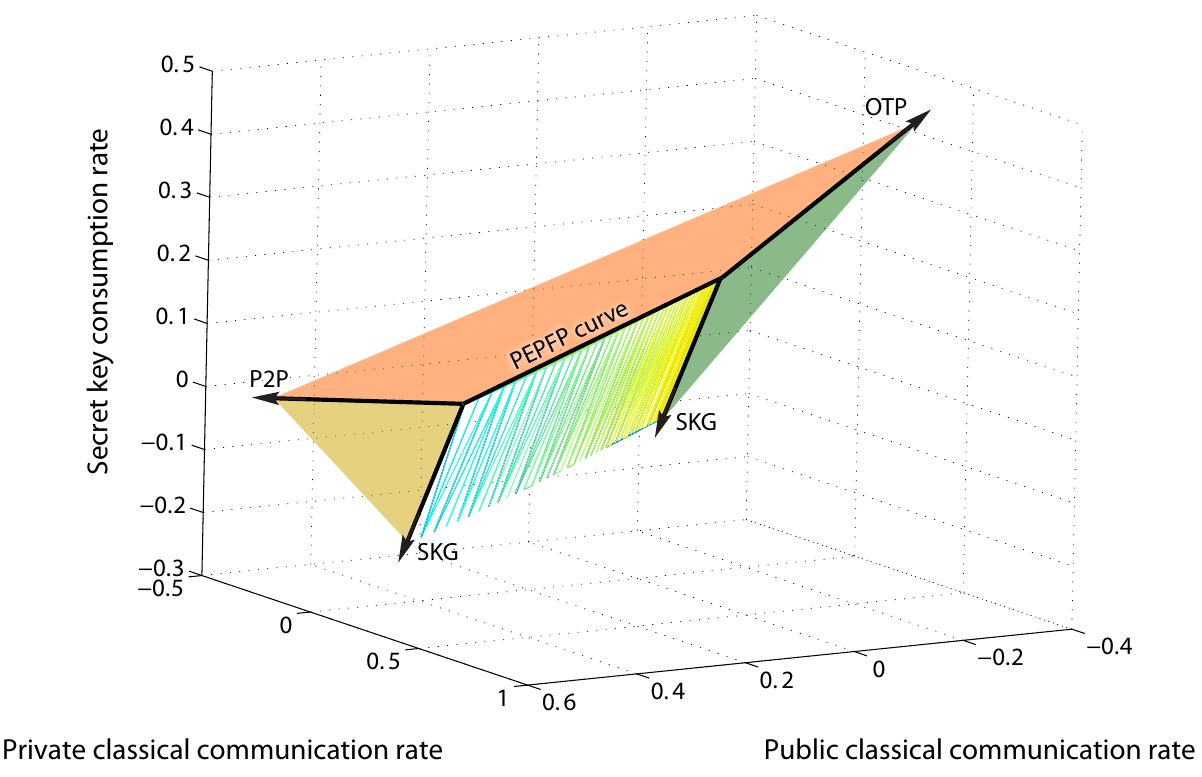}%
\caption{(Color online) The private dynamic capacity region for a
$1\rightarrow10$ cloning channel. The region exhibits a non-trivial resource
trade-off only on the surface below the PEPFP\ trade-off curve in the
direction of secret key distribution.}%
\label{fig:cloning-RPS}%
\end{center}
\end{figure}

\begin{theorem}
\label{thm:cloning-theorem}The private dynamic capacity region $\mathcal{C}%
_{\mathrm{{RPS}}}(\mathcal{N}_{\text{Cl}})$ of a $1\rightarrow N$ quantum
cloning channel is the set of all $R$, $P$, and $S$ such that%
\begin{align*}
R+P  &  \leq1-\log{N}+\frac{1}{\Delta_{N}}\sum_{i=0}^{N}i\log i,\\
P+S  &  \leq H\left(  \lambda_{i}\left(  \mu\right)  /\Delta_{N}\right)
-H\left(  \eta_{i}\left(  \mu\right)  /\Delta_{N}\right)  ,\\
R+P+S  &  \leq\log\left(  {N+1}\right)  -H\left(  \eta_{i}\left(  \mu\right)
/\Delta_{N}\right)  ,
\end{align*}
where $H$ is the entropy function $H\left(  \cdot\right)  \equiv-\sum
_{i}\left(  \cdot\right)  \log\left(  \cdot\right)  $,%
\begin{align*}
\Delta_{N}  &  \equiv N\left(  N+1\right)  /2,\\
\lambda_{i}\left(  \mu\right)   &  \equiv(N-2i)\mu+i\ \ \ \text{for\ \ \ }%
0\leq i\leq N,\\
\eta_{i}\left(  \mu\right)   &  \equiv(N-1-2i)\mu+i+1\ \ \ \text{for\ \ \ }%
0\leq i\leq N-1,\\
\mu &  \in\left[  0,1/2\right]  .
\end{align*}
\end{theorem}

\subsection{Quantum erasure channel}

\label{sec:single-letter-erasure}Below we show that the private dynamic
capacity region simplifies if the quantum channel is a quantum erasure
channel. A quantum erasure channel with erasure parameter $\epsilon$\ is the
following map:%
\[
\mathcal{N}_{\epsilon}\left(  \rho\right)  \equiv\left(  1-\epsilon\right)
\rho+\epsilon\left\vert e\right\rangle \left\langle e\right\vert .
\]
Notice that the receiver can perform a measurement $\left\{  \left\vert
0\right\rangle \left\langle 0\right\vert +\left\vert 1\right\rangle
\left\langle 1\right\vert ,\left\vert e\right\rangle \left\langle e\right\vert
\right\}  $ and can learn whether the channel erased the state. The receiver
can do this without disturbing the state in any way. An isometric extension
$U_{\mathcal{N}_{\epsilon}}^{A^{\prime}\rightarrow BE}$\ of it acts as follows
on a purification $\left\vert \psi\right\rangle ^{AA^{\prime}}$\ of the state
$\rho^{A^{\prime}}$:%
\[
U_{\mathcal{N}_{\epsilon}}^{A^{\prime}\rightarrow BE}\left\vert \psi
\right\rangle ^{AA^{\prime}}=\sqrt{1-\epsilon}\left\vert \psi\right\rangle
^{AB}\left\vert e\right\rangle ^{E}+\sqrt{\epsilon}\left\vert \psi
\right\rangle ^{AE}\left\vert e\right\rangle ^{B}.
\]
In the above representation, we see that the erasure channel has the
interpretation that it hands the input to Bob with probability $1-\epsilon$
while giving an erasure flag $\left\vert e\right\rangle $ to Eve, and it hands
the input to Eve with probability $\epsilon$ while giving the erasure flag to
Bob. Figure~\ref{fig:erasure-RPS}\ plots the region for an erasure channel
with erasure parameter $\epsilon=1/4$, and the proof of the below theorem
appears in the appendix.%
\begin{figure}
[ptb]
\begin{center}
\includegraphics[
natheight=4.880100in,
natwidth=8.186300in,
height=3.3174in,
width=5.5486in
]%
{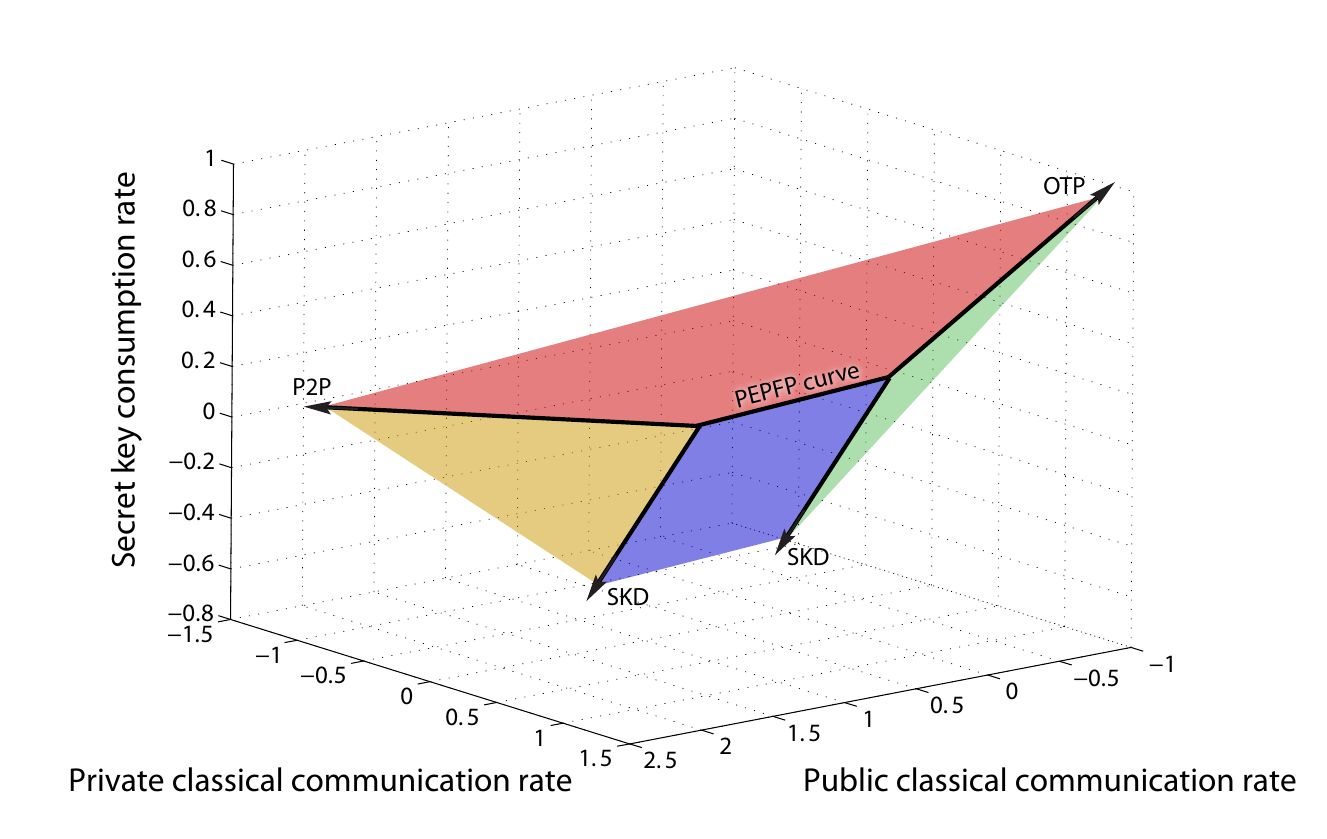}%
\caption{(Color online) The private dynamic capacity region for a quantum
erasure channel with erasure parameter $\epsilon=1/4$. The erasure channel
does not have a non-trivial trade-off, i.e., time-sharing between different
protocols is the optimal strategy.}%
\label{fig:erasure-RPS}%
\end{center}
\end{figure}

\begin{theorem}
\label{thm:erasure-theorem}The private dynamic capacity region $\mathcal{C}%
_{\mathrm{{RPS}}}(\mathcal{N}_{\epsilon})$ of a quantum erasure channel
$\mathcal{N}_{\epsilon}$ is the set of all $R$, $P$, and $S$ such that%
\begin{align*}
R+P  &  \leq\left(  1-\epsilon\right)  ,\\
P+S  &  \leq\left(  1-2\epsilon\right)  H_{2}\left(  p\right)  ,\\
R+P+S  &  \leq1-\epsilon-\epsilon H_{2}\left(  p\right)  ,
\end{align*}
where $p\in\left[  0,1/2\right]  $.
\end{theorem}

\section{Conclusion}

This paper completes the information-theoretic treatment of the
Collins-Popescu analogy between classical communication, quantum
communication, entanglement and public classical communication, private
classical communication, and secret key (at least for the case of channels).
Our main theorem gives the private dynamic capacity region of a quantum
channel. The catalytic information theoretic converse proof technique again
proves to be useful in obtaining a simplified converse proof. The private
dynamic capacity region dramatically simplifies for entanglement-breaking
channels, Hadamard channels, and erasure channels, so that we can actually
plot the region for several examples of these channels.

The open question remaining is to complete the Collins-Popescu analogy for the
case of a static resource (a bipartite state shared between Alice and Bob). We
have determined the static region for the classical-quantum-entanglement
trade-off~\cite{HW09T3}, and this first step should help in completing the
analogy. Another ambitious open question would be to solve the quintuple
trade-off between public classical communication, private classical
communication, quantum communication, entanglement, and secret key, of which
the regions in this paper are merely a projection. The catalytic
information-theoretic converse proof technique should be helpful in obtaining
a capacity theorem. Completing this larger trade-off problem could further our
understanding of the nature of these different resources and their interaction
with a noisy quantum resource.

\section*{Acknowledgements}

The authors thank Patrick Hayden for suggesting the communication model in
Figure~\ref{fig:comm-model}.

\appendix

\section{Proofs}

\begin{proof}
[Theorem~\ref{thm:dephasing} (Dephasing channel region)]We first
prove that it is sufficient to consider an ensemble of the following form to
characterize the boundary points of the region:%
\begin{multline}
\frac{\nu}{2}\left\vert 0\right\rangle \left\langle 0\right\vert ^{X}%
\otimes\left\vert 0\right\rangle \left\langle 0\right\vert ^{Y}\otimes
\left\vert 0\right\rangle \left\langle 0\right\vert ^{A^{\prime}}+\frac{1-\nu
}{2}\left\vert 0\right\rangle \left\langle 0\right\vert ^{X}\otimes\left\vert
1\right\rangle \left\langle 1\right\vert ^{Y}\otimes\left\vert 1\right\rangle
\left\langle 1\right\vert ^{A^{\prime}}+\\
\frac{1-\nu}{2}\left\vert 1\right\rangle \left\langle 1\right\vert ^{X}%
\otimes\left\vert 0\right\rangle \left\langle 0\right\vert ^{Y}\otimes
\left\vert 0\right\rangle \left\langle 0\right\vert ^{A^{\prime}}+\frac{\nu
}{2}\left\vert 1\right\rangle \left\langle 1\right\vert ^{X}\otimes\left\vert
1\right\rangle \left\langle 1\right\vert ^{Y}\otimes\left\vert 1\right\rangle
\left\langle 1\right\vert ^{A^{\prime}}, \label{eq:mu-cq-state-CEQ}%
\end{multline}
where $\nu\in\left[  0,1/2\right]  $. We can use the simplified form of the
private dynamic capacity formula in Lemma~\ref{lem:degradable-formula}%
\ because the dephasing channel is a degradable channel. Consider a
classical-quantum state with a finite number $N_{x}N_{y}$ of conditional
density operators $\phi_{x,y}^{A^{\prime}}$:%
\begin{equation}
\rho^{XYA^{\prime}}\equiv\sum_{x=0}^{N_{x}-1}\sum_{y=0}^{N_{y}-1}p_{X}\left(
x\right)  p_{Y|X}\left(  y|x\right)  |x\rangle\langle x|^{X}\otimes\left\vert
y\right\rangle \left\langle y\right\vert ^{Y}\otimes\phi_{x,y}^{A^{\prime}%
}.\nonumber
\end{equation}
Let $\phi_{x}^{A^{\prime}}$ denote the conditional states if $X$ is known but
$Y$ is not:%
\[
\phi_{x}^{A^{\prime}}\equiv\sum_{y=0}^{N_{y}-1}p_{Y|X}(y|x)\phi_{x,y}%
^{A^{\prime}}.
\]
It suffices for these states to be diagonal in the dephasing basis because the
channel output entropy when conditioned on $X$ can only be larger while the
environment's entropy when conditioned on $X$ remains constant (see Lemma~9 of
Ref.~\cite{itit2008hsieh}). We can form a new classical-quantum state with
quadruple the number of conditional density operators by applying all four
Pauli operators to the original conditional density operators:%
\[
\sigma^{XYJA^{\prime}}\equiv\sum_{x=0}^{N_{x}-1}\sum_{y=0}^{N_{y}-1}\sum
_{j=0}^{3}\frac{1}{4}p_{X}\left(  x\right)  p_{Y|X}\left(  y|x\right)
\ |x\rangle\langle x|^{X}\otimes|y\rangle\langle y|^{Y}\otimes|j\rangle\langle
j|^{J}\otimes\sigma_{j}\phi_{x,y}^{A^{\prime}}\sigma_{j},
\]
where $\sigma_{0}=I$, $\sigma_{1}=\sigma_{Z}$, $\sigma_{2}=\sigma_{X}$ and
$\sigma_{3}=\sigma_{Y}$ are the Pauli matrices. Let $\rho^{XYBE}$ and
$\sigma^{XYJBE}$ be the respective states after sending the $A^{\prime}$
system of $\rho^{XYA^{\prime}}$ and $\sigma^{XYJA^{\prime}}$ through the
isometric extension $U_{\mathcal{N}}^{A^{\prime}\rightarrow BE}$ of the
dephasing channel. Consider the following chain of inequalities that holds for
all $\lambda,\mu\geq0$:%
\begin{align*}
&  H\left(  B\right)  _{\rho}-H\left(  B|YX\right)  _{\rho}+\lambda\left[
H\left(  B|X\right)  _{\rho}-H\left(  E|X\right)  _{\rho}\right]  +\mu\left[
H\left(  B\right)  _{\rho}-H\left(  E|X\right)  _{\rho}\right] \\
&  =H\left(  B\right)  _{\rho}-H\left(  B|YXJ\right)  _{\sigma}+\lambda\left[
H\left(  B|XJ\right)  _{\sigma}-H\left(  E|XJ\right)  _{\sigma}\right]
+\mu\left[  H\left(  B\right)  _{\rho}-H\left(  E|XJ\right)  _{\sigma}\right]
\\
&  \leq H\left(  B\right)  _{\sigma}-H\left(  B|YXJ\right)  _{\sigma}%
+\lambda\left[  H\left(  B|XJ\right)  _{\sigma}-H\left(  E|XJ\right)
_{\sigma}\right]  +\mu\left[  H\left(  B\right)  _{\sigma}-H\left(
E|XJ\right)  _{\sigma}\right] \\
&  =1-H\left(  B|YXJ\right)  _{\sigma}+\lambda\left[  H\left(  B|XJ\right)
_{\sigma}-H\left(  E|XJ\right)  _{\sigma}\right]  +\mu\left[  1-H\left(
E|XJ\right)  _{\sigma}\right] \\
&  =1+\mu+\sum_{x=1}^{N_{x}-1}p_{X}\left(  x\right)  \left[  -H\left(
B\right)  _{\mathcal{N}\left(  \phi_{x,y}\right)  }+\lambda H\left(  B\right)
_{\mathcal{N}\left(  \phi_{x}\right)  }-\left(  \lambda+\mu\right)  H\left(
E\right)  _{\mathcal{N}^{c}\left(  \phi_{x}\right)  }\right] \\
&  \leq1+\mu+\max_{x}\left[  -H\left(  B\right)  _{\mathcal{N}\left(
\phi_{x,y}\right)  }+\lambda H\left(  B\right)  _{\mathcal{N}\left(  \phi
_{x}\right)  }-\left(  \lambda+\mu\right)  H\left(  E\right)  _{\mathcal{N}%
^{c}\left(  \phi_{x}\right)  }\right] \\
&  =1+\mu-H\left(  B\right)  _{\mathcal{N}\left(  \phi_{x,y}^{\ast}\right)
}+\lambda H\left(  B\right)  _{\mathcal{N}\left(  \phi_{x}^{\ast}\right)
}-\left(  \lambda+\mu\right)  H\left(  E\right)  _{\mathcal{N}^{c}\left(
\phi_{x}^{\ast}\right)  }\\
&  =1+\mu+\lambda H\left(  B\right)  _{\mathcal{N}\left(  \phi_{x}^{\ast
}\right)  }-\left(  \lambda+\mu\right)  H\left(  E\right)  _{\mathcal{N}%
^{c}\left(  \phi_{x}^{\ast}\right)  }\\
&  =1+\lambda\left[  H\left(  B\right)  _{\mathcal{N}\left(  \phi_{x}^{\ast
}\right)  }-H\left(  E\right)  _{\mathcal{N}^{c}\left(  \phi_{x}^{\ast
}\right)  }\right]  +\mu\left[  1-H\left(  E\right)  _{\mathcal{N}^{c}\left(
\phi_{x}^{\ast}\right)  }\right]  .
\end{align*}
The first equality follows because conditioning on $J$ does not change the
conditional entropies. That is, the conditional entropies $H\left(
B|X\right)  $ and $H\left(  B|YX\right)  $ are invariant under a Pauli
operator on the input state that commutes with the channel. Furthermore, a
Pauli operator on the input state does not change the eigenvalues for the
output of the dephasing channel's complementary channel: $H(E)_{\mathcal{N}%
^{c}(X\phi_{x}^{A^{\prime}}X)}=H(E)_{\mathcal{N}^{c}(\phi_{x}^{A^{\prime}})}$.
The first inequality follows because entropy is concave, i.e., the local state
$\sigma^{B}$ is a mixed version of $\rho^{B}$. The second equality follows
because%
\[
H(B)_{\sigma^{B}}=H\left(  \sum_{x,y,j}\frac{1}{4}p_{X}\left(  x\right)
p_{Y|X}\left(  y|x\right)  \sigma_{j}\phi_{x,y}^{B}\sigma_{j}\right)
=H\left(  \sum_{x,y}p_{X}\left(  x\right)  p_{Y|X}\left(  y|x\right)
I/2\right)  =1.
\]
The third equality follows because the system $X$ is classical and
conditioning on $J$ does not change the entropies. The second inequality
follows because the maximum value of a realization of a random variable is not
less than its expectation. The fourth equality follows be defining the
ensemble with a $\ast$ to be the optimal ensemble with respect to the
maximization over $x$. The fifth equality follows from a further optimization:
it is better to choose the pure states $\phi_{x,y}^{\ast}$ to be pure states
in the basis of the dephasing channel. The final equality follows by
rearranging terms. The final state $\phi_{x}^{\ast}$ then has the form
$\nu\left\vert 0\right\rangle \left\langle 0\right\vert ^{A^{\prime}}+\left(
1-\nu\right)  \left\vert 1\right\rangle \left\langle 1\right\vert ^{A^{\prime
}}$ for some value of $\nu$ because $\phi_{x}^{\ast}$ is diagonal in the
dephasing basis. The three other states $\sigma_{X}\phi_{x}^{\ast}\sigma_{X}$,
$\sigma_{Y}\phi_{x}^{\ast}\sigma_{Y}$, and $\sigma_{Z}\phi_{x}^{\ast}%
\sigma_{Z}$ have a similar form, but $\phi_{x}^{\ast}=\sigma_{Z}\phi_{x}%
^{\ast}\sigma_{Z}$ and $\sigma_{X}\phi_{x}^{\ast}\sigma_{X}=\sigma_{Y}\phi
_{x}^{\ast}\sigma_{Y}$. Thus, it suffices to choose the state $\phi_{x}^{\ast
}$ and its bit-flipped version, and the variable $Y$ needs only have
distribution $\left(  \nu,1-\nu\right)  $ because of the particular form of
$\phi_{x}^{\ast}$. Thus, an ensemble of the kind in (\ref{eq:mu-cq-state-CEQ})
is sufficient to attain a point on the boundary of the region. Evaluating the
entropic quantities in Theorem~\ref{thm:main-theorem}\ on a state of the above
form gives the expression for the region in Theorem~\ref{thm:dephasing}.
\end{proof}

\begin{proof}
[Theorem~\ref{thm:cloning-theorem} (Cloning channel region)]We first
prove that the same ensemble as in (\ref{eq:mu-cq-state-CEQ}) suffices for
achieving the limits of the region. We exploit the following classical-quantum
states:%
\begin{align*}
\rho^{XYA^{\prime}}  &  \equiv\sum_{x,y}p_{X}\left(  x\right)  p_{Y|X}\left(
y|x\right)  \left\vert x\right\rangle \left\langle x\right\vert ^{X}%
\otimes\left\vert y\right\rangle \left\langle y\right\vert ^{Y}\otimes
\phi_{x,y}^{A^{\prime}},\\
\sigma^{XYIA^{\prime}}  &  \equiv\sum_{x,i}\frac{1}{4}p_{X}\left(  x\right)
\left\vert x\right\rangle \left\langle x\right\vert ^{X}\otimes\left\vert
y\right\rangle \left\langle y\right\vert ^{Y}\otimes\left\vert i\right\rangle
\left\langle i\right\vert ^{I}\otimes(\sigma_{i}^{A^{\prime}})\phi
_{x,y}^{A^{\prime}}(\sigma_{i}^{A^{\prime}}),
\end{align*}
where the states $\phi_{x,y}^{A^{\prime}}$ are pure, and let $\rho^{XYBE}$ and
$\sigma^{XYIBE}$ be the states obtained by transmitting the $A^{\prime}$
system through the isometric extension of the erasure channel. Let $\sigma
_{x}^{A^{\prime}Y}\equiv\sum_{y}p_{Y|X}\left(  y|x\right)  \left\vert
y\right\rangle \left\langle y\right\vert ^{Y}\otimes\phi_{x,y}^{A^{\prime}}$.
The cloning channel is degradable and covariant~\cite{B09,BHP09}, the latter
meaning that the following relationships hold for any input density operator
$\sigma$ and any unitary $V$ acting on the input system $A^{\prime}$:%
\begin{align*}
\mathcal{N}_{\text{Cl}}\left(  V\sigma V^{\dag}\right)   &  =R_{V}%
\mathcal{N}_{\text{Cl}}\left(  \sigma\right)  R_{V}^{\dag},\\
\mathcal{N}_{\text{Cl}}^{c}\left(  V\sigma V^{\dag}\right)   &  =S_{V}%
\mathcal{N}_{\text{Cl}}^{c}\left(  \sigma\right)  S_{V}^{\dag},
\end{align*}
where $R_{V}$ and $S_{V}$ are higher-dimensional irreducible representations
of the unitary $V$ on the respective systems $B$ and $E$. The state
$\sigma^{B}$ is equal to the maximally mixed state on the symmetric subspace
for the following reasons:%
\begin{equation}
\sigma^{B}=\mathcal{N}_{\text{Cl}}\left(  \sigma^{A^{\prime}}\right)
=\mathcal{N}_{\text{Cl}}\left(  \frac{I^{A^{\prime}}}{2}\right)
=\mathcal{N}_{\text{Cl}}\left(  \int V\omega V^{\dag}\ \text{d}V\right)  =\int
R_{V}\mathcal{N}\left(  \omega\right)  R_{V^{\dag}}\ \text{d}V=\frac{1}%
{N+1}\sum_{i=0}^{N}\left\vert i\right\rangle \left\langle i\right\vert ^{B},
\label{eq:cloning-unital-relation}%
\end{equation}
where the fourth equality exploits the linearity and covariance of the cloning
channel $\mathcal{N}_{\text{Cl}}$.
Consider the following chain of inequalities:%
\begin{align*}
&  H\left(  B\right)  _{\rho}-H\left(  B|YX\right)  _{\rho}+\lambda\left[
H\left(  B|X\right)  _{\rho}-H\left(  E|X\right)  _{\rho}\right]  +\mu\left[
H\left(  B\right)  _{\rho}-H\left(  E|X\right)  _{\rho}\right] \\
&  =\left(  \mu+1\right)  H\left(  B\right)  _{\rho}-H\left(  B|YX\right)
_{\rho}+\lambda H\left(  B|X\right)  _{\rho}-\left(  \lambda+\mu\right)
H\left(  E|X\right)  _{\rho}\\
&  =\left(  \mu+1\right)  H\left(  B\right)  _{\rho}-H\left(  B|YXI\right)
_{\sigma}+\lambda H\left(  B|XI\right)  _{\sigma}-\left(  \lambda+\mu\right)
H\left(  E|XI\right)  _{\sigma}\\
&  \leq\left(  \mu+1\right)  H\left(  B\right)  _{\sigma}-H\left(
B|YXI\right)  _{\sigma}+\lambda H\left(  B|XI\right)  _{\sigma}-\left(
\lambda+\mu\right)  H\left(  E|XI\right)  _{\sigma}\\
&  =\left(  \mu+1\right)  \log\left(  N+1\right)  -\sum_{x,y}p_{X}\left(
x\right)  p_{Y|X}\left(  y|x\right)  H\left(  \frac{i}{\Delta_{N}}\right)
+\sum_{x}p_{X}\left(  x\right)  \left[  \lambda H\left(  B\right)
_{\mathcal{N}(\sigma_{x}^{A^{\prime}})}-\left(  \lambda+\mu\right)  H\left(
E\right)  _{\mathcal{N}^{c}(\sigma_{x}^{A^{\prime}})}\right] \\
&  \leq\left(  \mu+1\right)  \log\left(  N+1\right)  -H\left(  \frac{i}%
{\Delta_{N}}\right)  +\lambda H\left(  B\right)  _{\mathcal{N}(\sigma
_{x}^{\ast})}-\left(  \lambda+\mu\right)  H\left(  E\right)  _{\mathcal{N}%
^{c}(\sigma_{x}^{\ast})}\\
&  =1-\log{N}+\frac{1}{\Delta_{N}}\sum_{i=0}^{N}i\log i+\lambda\left[
H\left(  B\right)  _{\mathcal{N}(\sigma_{x}^{\ast})}-H\left(  E\right)
_{\mathcal{N}^{c}(\sigma_{x}^{\ast})}\right]  +\mu\left[  \log\left(
N+1\right)  -H\left(  E\right)  _{\mathcal{N}^{c}(\sigma_{x}^{\ast})}\right]
.
\end{align*}
The first equality follows by rearranging terms. The second equality follows
because the conditional entropies are invariant under unitary transformations:%
\[
H(B)_{R_{\sigma_{j}}\rho_{x}^{B}R_{\sigma_{j}}^{\dag}}=H(B)_{\rho_{x}^{B}%
},\ \ \ \ \ \ H(E)_{S_{\sigma_{j}}\rho_{x}^{E}S_{\sigma_{j}}^{\dag}%
}=H(E)_{\rho_{x}^{E}},
\]
where $R_{\sigma_{j}}$ and $S_{\sigma_{j}}$ are higher-dimensional
representations of $\sigma_{j}$ on systems $B$ and $E$, respectively. The
first inequality follows because entropy is concave, i.e., the local state
$\sigma^{B}$ is a mixed version of $\rho^{B}$. The third equality follows
because (\ref{eq:cloning-unital-relation}) implies that $H(B)_{\sigma^{B}%
}=\log\left(  N+1\right)  $, from applying unitary covariance of the cloning
channel to the term $H\left(  B|YXI\right)  _{\sigma}=\sum_{x,y}p_{X}\left(
x\right)  p_{Y|X}\left(  y|x\right)  H\left(  B\right)  _{\mathcal{N}\left(
\psi_{x,y}\right)  }$ (all pure states have the same output entropy---thus, it
does not matter which particular pure states we input), and from expanding the
conditional entropies $H\left(  B|XI\right)  _{\sigma}$ and $H\left(
E|XI\right)  _{\sigma}$. The second inequality follows because the maximum
value of a realization of a random variable is not less than its expectation.
The final equality follows by observing that $\log\left(  N+1\right)
-H\left(  \frac{i}{\Delta_{N}}\right)  =1-\log{N}+\frac{1}{\Delta_{N}}%
\sum_{i=0}^{N}i\log i$.
The entropies $H(B)_{\mathcal{N}(\sigma_{x}^{\ast})}$ and $H(E)_{\mathcal{N}%
^{c}(\sigma_{x}^{\ast})}$ depend only on the eigenvalues of the input state
$\sigma_{x}^{\ast}$ by the covariance of both the cloning channel and its
complement. We can therefore choose $\sigma_{x}^{\ast}$ to be diagonal in the
$\{|0\rangle,|1\rangle\}$ basis of $A^{\prime}$, and without loss of
generality, suppose these eigenvalues are equal to $\nu$ and $1-\nu$. The
ensemble defined to consist of $\sigma_{x}^{\ast}$ and $X\sigma_{x}^{\ast}X$
assigned equal probabilities then saturates the upper bound.
The final analytic form in the statement of the theorem follows by evaluating
the entropies and these calculations are similar to calculations available in
Section~V-B of Ref.~\cite{BHTW10}.
\end{proof}

\bigskip

\textbf{Proof of Theorem~\ref{thm:erasure-theorem} (Erasure channel region)}.
We prove Theorem~\ref{thm:erasure-theorem} in several steps.

\begin{lemma}
The private dynamic capacity formula in (\ref{eq:objective})\ simplifies as
follows for a quantum erasure channel~$\mathcal{N}_{\epsilon}$:%
\begin{equation}
P_{\lambda,\mu}\left(  \mathcal{N}_{\epsilon}\right)  \equiv\max_{p\in\left[
0,1/2\right]  }\left(  1-\epsilon\right)  +\lambda\left(  1-2\epsilon\right)
H_{2}\left(  p\right)  +\mu\left(  \left(  1-\epsilon\right)  -\epsilon
H_{2}\left(  p\right)  \right)  . \label{eq:objective-erasure}%
\end{equation}
Thus, the \textquotedblleft one-shot\textquotedblright\ dynamic capacity
region of a quantum erasure channel is as Theorem~\ref{thm:erasure-theorem} states.
\end{lemma}

\begin{proof}
We can use the simplified form of the private dynamic capacity formula in
Lemma~\ref{lem:degradable-formula}\ because the quantum erasure channel is a
degradable channel. We exploit the following classical-quantum states:%
\begin{align}
\rho^{XYA^{\prime}} &  \equiv\sum_{x,y}p_{X}\left(  x\right)  p_{Y|X}\left(
y|x\right)  \left\vert x\right\rangle \left\langle x\right\vert ^{X}%
\otimes\left\vert y\right\rangle \left\langle y\right\vert ^{Y}\otimes
\phi_{x,y}^{A^{\prime}},\nonumber\\
\sigma^{XYIA^{\prime}} &  \equiv\sum_{x,i}\frac{1}{4}p_{X}\left(  x\right)
\left\vert x\right\rangle \left\langle x\right\vert ^{X}\otimes\left\vert
y\right\rangle \left\langle y\right\vert ^{Y}\otimes\left\vert i\right\rangle
\left\langle i\right\vert ^{I}\otimes(\sigma_{i}^{A^{\prime}})\phi
_{x,y}^{A^{\prime}}(\sigma_{i}^{A^{\prime}}),\label{eq:cq-state-erasure-mixed}%
\end{align}
and let $\rho^{XYBE}$ and $\sigma^{XYIBE}$ be the states obtained by
transmitting the $A^{\prime}$ system through the isometric extension of the
erasure channel. Let $\sigma_{x}^{A^{\prime}}\equiv\sum_{y}p_{Y|X}\left(
y|x\right)  \phi_{x,y}^{A^{\prime}}$. Furthermore, let the eigenvalues of the
state $\sigma_{x}^{A^{\prime}}$ with highest entropy on system $A^{\prime}$ be
$p$ and $1-p$. Consider that the following chain of inequalities holds for any
state $\rho^{XYBE}$:%
\begin{align*}
&  H\left(  B\right)  _{\rho}-H\left(  B|YX\right)  _{\rho}+\lambda\left[
H\left(  B|X\right)  _{\rho}-H\left(  E|X\right)  _{\rho}\right]  +\mu\left[
H\left(  B\right)  _{\rho}-H\left(  E|X\right)  _{\rho}\right]  \\
&  =\left(  \mu+1\right)  H\left(  B\right)  _{\rho}-H\left(  B|YX\right)
_{\rho}+\lambda H\left(  B|X\right)  _{\rho}-\left(  \lambda+\mu\right)
H\left(  E|X\right)  _{\rho}\\
&  =\left(  \mu+1\right)  H\left(  B|X_{E}\right)  _{\rho}-H\left(
B|YXX_{E}\right)  _{\rho}+\lambda H\left(  B|XX_{E}\right)  _{\rho}-\left(
\lambda+\mu\right)  H\left(  E|XX_{E}\right)  _{\rho}\\
&  =\left(  \mu+1\right)  \left(  1-\epsilon\right)  H\left(  A^{\prime
}\right)  _{\rho}-\left(  1-\epsilon\right)  H\left(  A^{\prime}|YX\right)
_{\rho}+\lambda\left(  1-\epsilon\right)  H\left(  A^{\prime}|X\right)
_{\rho}-\left(  \lambda+\mu\right)  \epsilon H\left(  A^{\prime}|X\right)
_{\rho}\\
&  =\left(  \mu+1\right)  \left(  1-\epsilon\right)  H\left(  A^{\prime
}\right)  _{\rho}+\left[  \lambda\left(  1-\epsilon\right)  -\left(
\lambda+\mu\right)  \epsilon\right]  H\left(  A^{\prime}|X\right)  _{\rho}\\
&  =\left(  \mu+1\right)  \left(  1-\epsilon\right)  H\left(  A^{\prime
}\right)  _{\rho}+\left[  \lambda\left(  1-\epsilon\right)  -\left(
\lambda+\mu\right)  \epsilon\right]  H\left(  A^{\prime}|XI\right)  _{\sigma}.
\end{align*}
The first equality follows by rearrnging terms. The second equality follows by
incorporating the classical erasure flag variable. The third equality follows
by exploiting the properties of the quantum erasure channel. The fourth
equality follows by rearranging terms and because the entropy $H\left(
A^{\prime}|YX\right)  _{\rho}$ vanishes (the state on $A^{\prime}$ conditioned
on both $X$ and $Y$ is pure). The fifth equality follows because $H\left(
A^{\prime}|X\right)  _{\rho}=H\left(  A^{\prime}|XI\right)  _{\sigma}$.
Continuing,%
\begin{align*}
&  \leq\left(  \mu+1\right)  \left(  1-\epsilon\right)  H\left(  A^{\prime
}\right)  _{\sigma}+\left[  \lambda\left(  1-\epsilon\right)  -\left(
\lambda+\mu\right)  \epsilon\right]  H\left(  A^{\prime}|XI\right)  _{\sigma
}\\
&  =\left(  \mu+1\right)  \left(  1-\epsilon\right)  +\left[  \lambda\left(
1-\epsilon\right)  -\left(  \lambda+\mu\right)  \epsilon\right]  H\left(
A^{\prime}|XI\right)  _{\sigma}\\
&  =\left(  \mu+1\right)  \left(  1-\epsilon\right)  +\left[  \lambda\left(
1-\epsilon\right)  -\left(  \lambda+\mu\right)  \epsilon\right]  \sum_{x}%
p_{X}\left(  x\right)  H\left(  A^{\prime}\right)  _{\sigma_{x}^{A^{\prime}}%
}\\
&  \leq\left(  \mu+1\right)  \left(  1-\epsilon\right)  +\left[
\lambda\left(  1-\epsilon\right)  -\left(  \lambda+\mu\right)  \epsilon
\right]  H\left(  A^{\prime}\right)  _{\sigma_{x}^{\ast}}\\
&  =\left(  1-\epsilon\right)  +\lambda\left(  1-2\epsilon\right)
H_{2}\left(  p\right)  +\mu\left(  \left(  1-\epsilon\right)  -\epsilon
H_{2}\left(  p\right)  \right)  .
\end{align*}
The first inequality follows because the unconditional entropy of the state
$\rho$ is always less than that of the state $\sigma$. The first equality
follows because $H\left(  A^{\prime}\right)  _{\sigma}=1$. The second equality
follows by expanding the conditional entropy. The second inequality follows
because an average is always less than a maximum. The final equality follows
by rearranging terms and by plugging in the eigenvalues of $\sigma_{x}^{\ast}%
$. The form of the private dynamic capacity formula then follows because this
chain of inequalities holds for any input ensemble.
\end{proof}

\begin{lemma}
\label{lem:suff-condition}It suffices to consider the set of $\lambda,\mu
\geq0$ for which%
\[
\lambda\left(  1-2\epsilon\right)  \geq\mu\epsilon.
\]
Otherwise, we are just maximizing the public classical capacity, which we know
from Ref.~\cite{PhysRevLett.78.3217} is equal to $1-\epsilon$.
\end{lemma}

\begin{proof}
Consider rewriting the expression in (\ref{eq:objective-erasure}) as follows:%
\[
\max_{p\in\left[  0,1/2\right]  }\left(  1-\epsilon\right)  +\mu\left(
1-\epsilon\right)  +\left[  \lambda\left(  1-2\epsilon\right)  -\mu
\epsilon\right]  H_{2}\left(  p\right)  .
\]
Suppose that the expression in square brackets is negative, i.e.,%
\[
\lambda\left(  1-2\epsilon\right)  <\mu\epsilon.
\]
Then the maximization over $p$ simply chooses $p=0$ so that $H_{2}\left(
p\right)  $ vanishes and the negative term disappears. The resulting
expression for the private dynamic capacity formula is%
\[
\left(  1-\epsilon\right)  +\mu\left(  1-\epsilon\right)  ,
\]
which corresponds to the following region%
\begin{align*}
R+P  &  \leq1-\epsilon,\\
P+S  &  \leq0,\\
R+P+S  &  \leq1-\epsilon.
\end{align*}
The above region is equivalent to a translation of the unit resource capacity
region to the public classical capacity rate triple $\left(  1-\epsilon
,0,0\right)  $. Thus, it suffices to restrict the parameters $\lambda$ and
$\mu$ as above for the quantum erasure channel.
\end{proof}

\begin{lemma}
\label{lem:erasure-base-case}The following additivity relation holds for two
quantum erasure channels $\mathcal{N}_{\epsilon}$ with the same erasure
parameter $\epsilon$:%
\[
P_{\lambda,\mu}(\mathcal{N}_{\epsilon}\otimes\mathcal{N}_{\epsilon
})=P_{\lambda,\mu}(\mathcal{N}_{\epsilon})+P_{\lambda,\mu}(\mathcal{N}%
_{\epsilon}).
\]
\end{lemma}

\begin{proof}
We prove the non-trivial inequality $P_{\lambda,\mu}(\mathcal{N}_{\epsilon
}\otimes\mathcal{N}_{\epsilon})\leq P_{\lambda,\mu}(\mathcal{N}_{\epsilon
})+P_{\lambda,\mu}(\mathcal{N}_{\epsilon})$. We define the following states:%
\begin{align*}
\rho^{XYA_{1}^{\prime}A_{2}^{\prime}}  &  \equiv\sum_{x,y}p_{X}\left(
x\right)  p_{Y|X}\left(  y|x\right)  \left\vert x\right\rangle \left\langle
x\right\vert ^{X}\otimes\left\vert y\right\rangle \left\langle y\right\vert
^{Y}\otimes\phi_{x,y}^{A_{1}^{\prime}A_{2}^{\prime}},\\
\omega^{XYB_{1}E_{1}B_{2}E_{2}}  &  \equiv U_{\mathcal{N}_{\epsilon}}%
^{A_{1}^{\prime}\rightarrow B_{1}E_{1}}\otimes U_{\mathcal{N}_{\epsilon}%
}^{A_{2}^{\prime}\rightarrow B_{2}E_{2}}(\rho^{XYA_{1}^{\prime}A_{2}^{\prime}%
}),
\end{align*}
and we suppose that $\rho^{XYA_{1}A_{2}}$ is the state that maximizes
$P_{\lambda,\mu}(\mathcal{N}_{\epsilon}\otimes\mathcal{N}_{\epsilon})$.
Consider the following equality:%
\begin{align*}
&  H\left(  B_{1}B_{2}\right)  _{\omega}-H\left(  B_{1}B_{2}|YX\right)
_{\omega}+\lambda\left(  H\left(  B_{1}B_{2}|X\right)  _{\omega}-H\left(
E_{1}E_{2}|X\right)  _{\omega}\right)  +\mu\left(  H\left(  B_{1}B_{2}\right)
_{\omega}-H\left(  E_{1}E_{2}|X\right)  _{\omega}\right) \\
&  =\left(  1-\epsilon\right)  ^{2}H\left(  A_{1}^{\prime}A_{2}^{\prime
}\right)  _{\rho}+\epsilon\left(  1-\epsilon\right)  \left(  H\left(
A_{1}^{\prime}\right)  _{\rho}+H\left(  A_{2}^{\prime}\right)  _{\rho}\right)
\\
&  \ \ \ \ \ \ -\epsilon\left(  1-\epsilon\right)  \left(  H\left(
A_{1}^{\prime}|YX\right)  _{\rho}+H\left(  A_{2}^{\prime}|YX\right)  _{\rho
}\right) \\
&  \ \ \ \ \ \ +\lambda\left[  \left(  1-\epsilon\right)  ^{2}H\left(
A_{1}^{\prime}A_{2}^{\prime}|X\right)  _{\rho}+\epsilon\left(  1-\epsilon
\right)  \left(  H\left(  A_{1}^{\prime}|X\right)  _{\rho}+H\left(
A_{2}^{\prime}|X\right)  _{\rho}\right)  \right] \\
&  \ \ \ \ \ \ -\lambda\left[  \epsilon^{2}H\left(  A_{1}^{\prime}%
A_{2}^{\prime}|X\right)  _{\rho}+\epsilon\left(  1-\epsilon\right)  \left(
H\left(  A_{1}^{\prime}|X\right)  _{\rho}+H\left(  A_{2}^{\prime}|X\right)
_{\rho}\right)  \right] \\
&  \ \ \ \ \ \ +\mu\left[  \left(  1-\epsilon\right)  ^{2}H\left(
A_{1}^{\prime}A_{2}^{\prime}\right)  _{\rho}+\epsilon\left(  1-\epsilon
\right)  \left(  H\left(  A_{1}^{\prime}\right)  _{\rho}+H\left(
A_{2}^{\prime}\right)  _{\rho}\right)  \right] \\
&  \ \ \ \ \ \ -\mu\left[  \epsilon^{2}H\left(  A_{1}^{\prime}A_{2}^{\prime
}|X\right)  _{\rho}+\epsilon\left(  1-\epsilon\right)  \left(  H\left(
A_{1}^{\prime}|X\right)  _{\rho}+H\left(  A_{2}^{\prime}|X\right)  _{\rho
}\right)  \right]  .
\end{align*}
The above equality follows by exploiting the properties of the quantum erasure
channel and because the entropy $H\left(  A_{1}^{\prime}A_{2}^{\prime
}|YX\right)  _{\rho}=0$. Continuing, the above quantity is less than the
following one:%
\begin{align*}
&  \leq2\left(  1-\epsilon\right)  -\epsilon\left(  1-\epsilon\right)  \left(
H\left(  A_{1}^{\prime}|YXIJ\right)  _{\sigma}+H\left(  A_{2}^{\prime
}|YXIJ\right)  _{\sigma}\right) \\
&  \ \ \ \ \ \ +\lambda\left(  1-2\epsilon\right)  H\left(  A_{1}^{\prime
}A_{2}^{\prime}|XIJ\right)  _{\sigma}\\
&  \ \ \ \ \ \ +\mu\left[  2\left(  1-\epsilon\right)  -\epsilon^{2}H\left(
A_{1}^{\prime}A_{2}^{\prime}|XIJ\right)  _{\sigma}-\epsilon\left(
1-\epsilon\right)  \left(  H\left(  A_{1}^{\prime}|XIJ\right)  _{\sigma
}-H\left(  A_{2}^{\prime}|XIJ\right)  _{\sigma}\right)  \right] \\
&  =2\left(  1-\epsilon\right) \\
&  \ \ \ \ \ \ \ +\lambda\left(  1-2\epsilon\right)  \left(  H\left(
A_{1}^{\prime}|XIJ\right)  _{\sigma}+H\left(  A_{2}^{\prime}|XIJ\right)
_{\sigma}\right) \\
&  \ \ \ \ \ \ +\mu\left[  2\left(  1-\epsilon\right)  -\epsilon\left(
H\left(  A_{1}^{\prime}|XIJ\right)  _{\sigma}-H\left(  A_{2}^{\prime
}|XIJ\right)  _{\sigma}\right)  \right] \\
&  \ \ \ \ \ \ \ -\epsilon\left(  1-\epsilon\right)  \left(  H\left(
A_{1}^{\prime}|YXIJ\right)  _{\sigma}+H\left(  A_{2}^{\prime}|YXIJ\right)
_{\sigma}\right) \\
&  \ \ \ \ \ \ -\left[  \left(  \lambda\left(  1-2\epsilon\right)
-\mu\epsilon^{2}\right)  I\left(  A_{1}^{\prime};A_{2}^{\prime}|XIJ\right)
_{\sigma}\right] \\
&  \leq P_{\lambda,\mu}\left(  \mathcal{N}_{\epsilon}\right)  +P_{\lambda,\mu
}\left(  \mathcal{N}_{\epsilon}\right)  -\epsilon\left(  1-\epsilon\right)
\left(  H\left(  A_{1}^{\prime}|YXIJ\right)  _{\sigma}+H\left(  A_{2}^{\prime
}|YXIJ\right)  _{\sigma}\right) \\
&  \ \ \ \ \ \ \ -\left[  \left(  \lambda\left(  1-2\epsilon\right)
-\mu\epsilon^{2}\right)  I\left(  A_{1}^{\prime};A_{2}^{\prime}|XIJ\right)
_{\rho}\right] \\
&  \leq P_{\lambda,\mu}\left(  \mathcal{N}_{\epsilon}\right)  +P_{\lambda,\mu
}\left(  \mathcal{N}_{\epsilon}\right)  .
\end{align*}
The first inequality follows from similar proofs we have seen for a state
$\sigma$\ of the form in (\ref{eq:cq-state-erasure-mixed}). The first equality
follows by rearranging terms. The second inequality follows from the form of
$D_{\lambda,\mu}$ in (\ref{eq:objective-erasure}). The final inequality
follows because Lemma~\ref{lem:suff-condition} states that it is sufficient to
consider $\lambda\left(  1-2\epsilon\right)  \geq\mu\epsilon$. Note that this
condition implies that%
\[
\lambda\left(  1-2\epsilon\right)  \geq\mu\epsilon^{2},
\]
and hence that the quantity in square brackets in the line above the last one
is positive.
\end{proof}

\end{document}